\documentclass[10pt,a4paper]{article}

\usepackage{lmodern}
\usepackage[latin1]{inputenc}
\usepackage{graphicx,color}
\usepackage{epsfig}
\usepackage{verbatim}
\usepackage[english]{babel}
\usepackage{amsmath, amsthm, amssymb}
\usepackage{wrapfig}
\usepackage{neubauer}

\usepackage{placeins}
\usepackage{appendix}
\usepackage{csquotes}

\newtheorem{prop}{Proposition}
\newtheorem{theorem}{Theorem}
\newtheorem{definition}{Definition}
\newtheorem{lemma}{Lemma}
\newtheorem{remark}{Remark}

\DeclareMathOperator{\sinc}{sinc}
 \DeclareMathOperator{\sgn}{sgn}

\makeatletter 
\renewcommand{\@fnsymbol}[1]{\@arabic{#1}}
\makeatother

\author{Victoria Hutterer\thanks{Industrial Mathematics Institute, Johannes Kepler University, Linz, Austria. victoria.hutterer@indmath.uni-linz.ac.at}, Ronny Ramlau\footnote{Industrial Mathematics Institute, Johannes Kepler University, Linz, and Johann Radon Institute for Computational and Applied Mathematics, Linz,  Austria.} \ and Iuliia Shatokhina\footnote{Johann Radon Institute for Computational and Applied Mathematics, Linz,  Austria.}}
\title{Real-time Adaptive Optics with pyramid wavefront sensors: \\ A theoretical analysis of the pyramid sensor model }

\begin{document}

\maketitle


\begin{abstract}
We consider the mathematical background of the wavefront sensor type that is widely used in Adaptive Optics systems for astronomy, microscopy, and ophthalmology. The theoretical analysis of the pyramid sensor forward operators presented in this paper is aimed at a subsequent development of fast and stable algorithms for wavefront reconstruction from data of this sensor type. In our analysis we allow the sensor to be utilized in both the modulated and non-modulated fashion. We derive detailed mathematical models for the pyramid sensor and the physically simpler roof wavefront sensor as well as their various approximations. Additionally, we calculate adjoint operators which build preliminaries for the application of several iterative mathematical approaches for solving inverse problems such as gradient based algorithms, Landweber iteration or Kaczmarz methods. 
\end{abstract}

\section{Introduction}
Ground-based telescope facilities suffer from degraded image quality caused by atmospheric turbulence. When light from a distant star passes the Earth's atmosphere, initially planar wavefronts get distorted due to turbulent air motions causing fluctuations of the index of refraction. Therefore, advanced Adaptive Optics (AO) systems \cite{Ha98, Roddier} are incorporated in innovative telescope systems to mechanically correct in real-time for the distortions with deformable mirrors. The shape of the deformable mirrors is determined by measuring wavefronts coming from either bright astronomical stars or artificially produced laser beacons.
The basic idea is to reflect the distorted wavefronts on a mirror  that is shaped appropriately such that the corrected wavefronts allow for high image quality when observed by the science camera (see Figure \ref{fig:dm}). The according positioning of the mirror actuators implies the knowledge of the incoming wavefronts. Thus, in Adaptive Optics one is interested in the reconstruction of the unknown incoming wavefront $\Phi$ from available data in order to calculate the optimal shape of the deformable mirror. Unfortunately, there exists no optical device which is able to measure the wavefront directly. Instead, a wavefront sensor (WFS) measures the time-averaged characteristic of the captured light that is related to the incoming phase. The wavefront sensor splits the telescope aperture into many small, equally spaced subapertures and detects the intensity of the incoming light in each of the subapertures. The number of subapertures defines the spatial resolution of the AO system. The detector of the sensor provides the intensity data by integrating the light from the celestial object. From the detector's analog signal, digital sensor measurements are sampled with a time period $T$ which is usually extremely short, i.e., 0.3~-~2~milliseconds. \bigskip

\begin{figure}[!ht]
\centering
\includegraphics[scale=0.3]{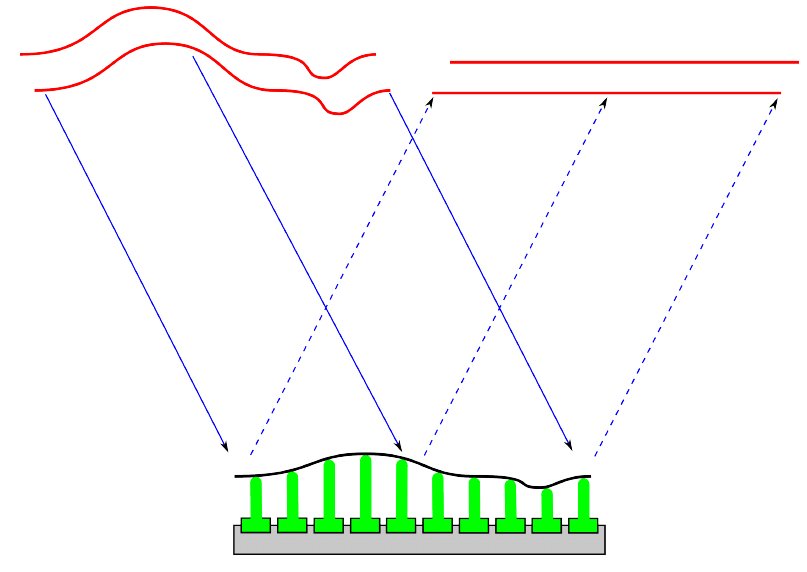}
\caption{Principle of wavefront correction \cite{Au17}. A deformable mirror reflects a perturbed wavefront and propagates the corrected, planar wavefront to the science camera yielding improved image quality.}
\label{fig:dm}
\end{figure}

This paper is focused on the pyramid wavefront sensor (PWFS), which has been invented in the 1990s \cite{Raga96}. It is gradually gaining more and more attention from the scientific community, especially in astronomical AO, due to its increased sensitivity, improved signal-to-noise ratio, robustness to spatial aliasing and adjustable spatial sampling compared to the other popular wavefront sensor choice --- the Shack-Hartmann (SH) sensor. Several theoretical studies ~\cite{BuDa06,EspoRi01,EsRiFe00,KoVe07,RaDi02,RaFa99,Shatokhina_PhDThesis,Veran_ao4elt4_pwfs_vs_sh,
Verinaud_2005_pwfs_vs_sh,Veri04,Viotto_2016_pwfs_gain} including numerical simulations and laboratory investigations with optical test benches~\cite{Bond_2016,Fusco_ao4elt3_pwfs_lab,Martin_ao4elt4_pwfs_bench_on-sky,Pinna_2008_HOT_bench,Turbide_ao4elt3_pwfs_bench,Ragazzoni_ao4elt3_pwfs_lab}  have confirmed the advantages of pyramid wavefront sensors  while additionally promising surveys were operated on sky~\cite{Esposito_2011_pwfs_onsky,Esposito_ao4elt2_LBT_onsky,Esposito2010_AO_for_LBT,FePeHi06,
Ghedina_2003_first_pwfs_onsky,Pedichini_2016_XAO_pwfs_LBT,PeFeHe10,Pinna_ao4elt4_XAO_LBT_onsky}. 

The current development of a new era of Extremely Large Telescopes (ELTs) with primary mirrors of $25-40$~m in diameter brings new challenges to the field of Adaptive Optics. For ELTs, pyramid sensors show enhanced performance in real life settings, e.g., they provide the ability to sense differential piston modes induced by diffraction effects of realistic telescope spiders that support secondary mirrors and perform even under significant levels of non-common path aberrations~\cite{VanDam_2012_pwfs_truth_LTAO_GMT,Esposito_2012_pwfs_NGS_SCAO_GMT}. 

Pyramid wavefront sensors are going to be included in many ELT instruments ~\cite{METIS_archiv,Clenet_SPIE_2016_micado_pwfs,VanDam_2012_pwfs_truth_LTAO_GMT,Esposito_2012_pwfs_NGS_SCAO_GMT,
Fusco_2010_spie_harmoni_eelt,Fusco_2010_spie_atlas_eelt,KoVe10,Macintosh_2006_spie_pwfs_xao_tmt,
Mieda_spie_2016_pwfs_truth_TMT,Neichel_2016_spie_pwfs_harmoni_eelt,Veran_ao4elt4_pwfs_vs_sh}. We would like to mention that, apart from astronomical applications, the pyramid wavefront sensor is also applied in adaptive loops in ophthalmology~\cite{Chamot06,Daly_2010,Alvarez_Thesis,Iglesias02} and microscopy~\cite{Ig11,Ig13}. Therefore, wavefront reconstruction algorithms for pyramid wavefront sensors are in high demand. The main goal of this paper is to provide an extensive mathematical analysis of the PWFS operators in order to develop suitable wavefront reconstruction methods. So far, the existing algorithms are (with a small number of exceptions as, e.g., \cite{Clare_2004,deo_spie_2018,Frazin_18,frazin_spie2018,Hut18,korkiakoski_nonlinear_08,Korkiakoski_08,
KoVe07,Viotto16}) based on a linear assumption of the pyramid sensor model \cite{Hu18_thesis,ShatHut_spie2018_overview}. Nevertheless, the simplifications of the non-linear pyramid operator allow for acceptable wavefront reconstruction quality. \bigskip

Basically,  wavefront reconstruction from pyramid sensor data consists in solving two non-linear integral equations 
\begin{equation*}
\boldsymbol P\Phi = [s_x,s_y] 
\end{equation*}
with respect to the unknown wavefront $\Phi$, where $\boldsymbol P = \left[\boldsymbol P_x,\boldsymbol P_y\right]$ is a singular Volterra integral operator of the first kind. \bigskip

In the following, we derive and analyze the mathematical model for both the non- and modulated pyramid sensor. 

The paper is organized as follows: In Section \ref{chap:pwfs_intro}, we give a brief introduction into the physical background of pyramid and roof wavefront sensors. Afterwards, we derive the singular and non-linear forward PWFS operator $\boldsymbol P$ and roof sensor operators $\boldsymbol R$. In order to simplify the problem of solving the WFS equation, we calculate linearizations $\boldsymbol R^{lin}$ for roof wavefront sensors in Section \ref{chap:pwfs_forward_models}. Furthermore, in Section \ref{chap:adjoint_operators}, we evaluate the corresponding adjoint operators which are necessary for the application of linear iterative methods such as gradient based algorithms or Landweber iteration for wavefront reconstruction addressed in an upcoming second part of the paper.

\section{Pyramid and roof wavefront sensor}\label{chap:pwfs_intro}

\begin{figure}[!ht]
\centering
\includegraphics[scale=0.25]{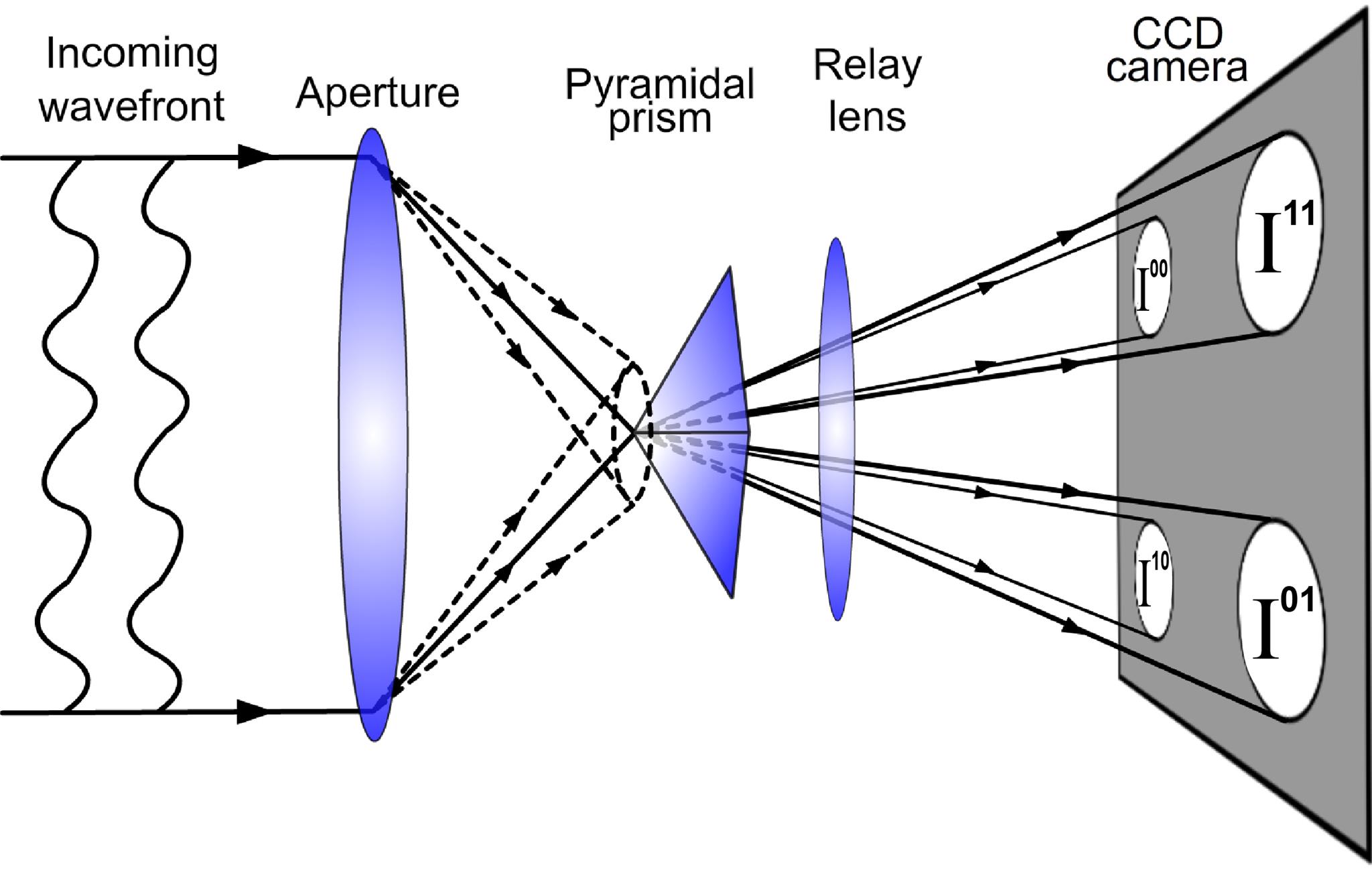}
\caption{Scheme of the optical setup of a pyramid WFS. The circular modulation path is shown in the dashed line. }
\label{fig:Pyramid_sensor_AP}
\end{figure}

As it can be seen in Figure \ref{fig:Pyramid_sensor_AP}, the main component of the pyramid sensor is a four-sided glass pyramidal 
prism placed in the focal plane of the telescope pupil. The incoming light is focused by the telescope onto the prism apex. 
The four facets of the pyramid split the incoming light in four beams, propagated in slightly different directions. 
A relay lens, placed behind the prism, re-images the four beams, allowing adjustable sampling of the four different images $I_{ij}$, $i,j=\{1,2\}$
of the aperture on the CCD camera. The two measurement sets $s_x$, $s_y$ are obtained from the four intensity 
patterns as 

\begin{equation}\label{meas_equ}
\begin{split}
s_x(x,y) = \dfrac{ [I^{01}(x,y) + I^{00}(x,y)] - [I^{11}(x,y) + I^{10}(x,y)] } {I_0} , \\
s_y(x,y) = \dfrac{ [I^{01}(x,y) + I^{11}(x,y)] - [I^{00}(x,y) + I^{10}(x,y)] } {I_0} ,
\end{split}
\end{equation} 
where $I_0$ is the average intensity per subaperture. 

For the sake of simplicity, we focus on the transmission mask modeling approach, which ignores the phase shifts introduced by the pyramidal prism and models the prism facets as transmitting only. The interference effects are neglected as well.

A dynamic circular modulation of the incoming beam allows to increase the linear range of the pyramid sensor~\cite{Veri04} and is also used to adjust its sensitivity. The modulation can be accomplished in several ways: either by oscillating the pyramid itself~\cite{Raga96}, with a steering \mbox{mirror~\cite{BuDa06,LeDueJoli}}, or by using a static diffusive optical element~\cite{LeDueJoli}. The circular modulation path of the focused beam on the pyramid apex is shown with a dashed circle in Figure~\ref{fig:Pyramid_sensor_AP}. 

We consider the roof wavefront sensor as a stand-alone WFS and as a simplification of the pyramid sensor. For this type of sensor, two orthogonally placed two-sided roof prisms instead of the pyramidal one are used. The orthogonality of the roof leads to a decoupling of the dependence of the sensor data on the incoming phase in $x$- and $y$-direction. For roof wavefront sensors, a linear modulation path is additionally considered in the literature as an approximation to the circular one.

The understanding of the physics behind the pyramid sensor has changed over time. In the beginning, the PWFS was introduced 
with dynamic modulation of the incoming beam and described within the geometric optics framework as a slope sensor 
similar to SH sensors, but having a higher sensitivity~\cite{Esposito_00_labtest,Raga96,RaFa99,Verinaud_2005_pwfs_vs_sh}. 

Later, the role of the beam modulation was questioned and the pyramid sensor without modulation was \mbox{studied} as well~\cite{Korkiakoski_08,Lee_00,RaDi02}. According to the Fourier optics based analytical model derived in~\cite{KoVe07}, the non-modulated PWFS measures a non-linear combination of one and two dimensional Hilbert transforms of the sine and cosine of the incoming distorted wavefront. 

Then, it was recognized that the dynamic modulation of the beam allows to strengthen the linearity of the sensor and to increase 
its dynamic range. Taking modulation into account within the Fourier optics framework complicates the  
non-linear forward model even more. However, a linearization of both models, with and without modulation, is possible under certain simplifying assumptions~\cite{BuDa06,Veri04}. Within the linearized model, it was shown that the modulated pyramid sensor measures both the slope and the Hilbert transform of the wavefront, depending on the frequency range and the amount of modulation~\cite{Veri04}. \bigskip

The full continuous measurements $s_x(x,y)$ and $s_y(x,y)$ of the pyramid wavefront sensor are not available in practice. For the description of the discrete pyramid sensor we perform a division of the continuous two dimensional process into finitely many equispaced regions called subapertures. The data are then assumed to be averaged over every subaperture which corresponds to the finite sampling of the pyramid sensor. Hence, the data grid is predetermined by a subaperture size of $d \cdot d$ with $d=\frac{D}{n}$, where $D$ represents the telescope diameter, i.e., the primary mirror size, and $n$ the number of subapertures in one direction. 

\begin{figure}[!ht]
\centering
\includegraphics[scale=0.5]{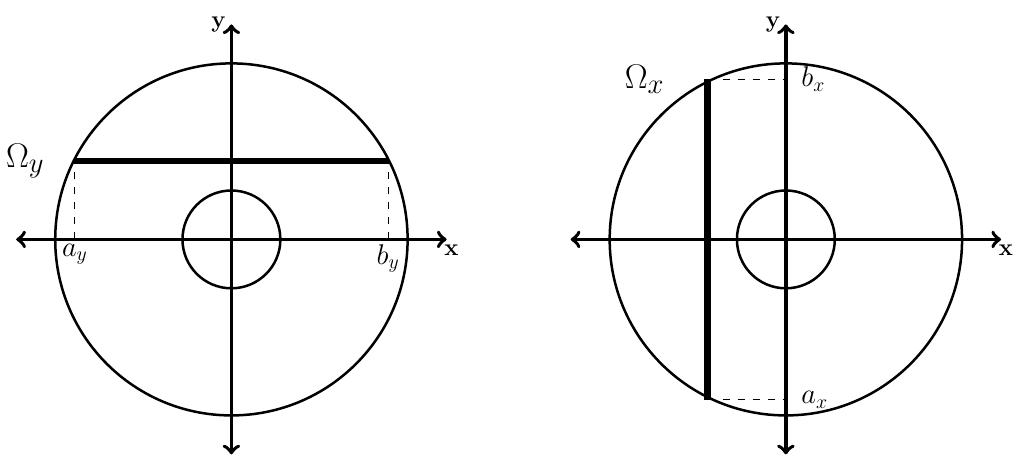}
\caption{The figures show the borders of the annular aperture mask for fixed $x$ and $y$. The domain $\Omega_x$ changes with $x$ and $\Omega_y$ with $y$ respectively. In some cases the intervals are split due to the central obstruction of the telescope.}
\label{fig:omega}
\end{figure}
Additionally, we restrict the availability of measurements to the size of the region captured by the sensor. For several telescope systems the pupil is annular instead of circular since a secondary mirror shades the primary mirror, making the area of central obstruction hardly attainable for photons. Thus, the remaining light in the area of the central obstruction does not produce reliable measurements. Moreover, the incoming phases are defined on $\R^2$ but for the control of the deformable mirror, we are only interested in the reconstructed wavefront shape on a restricted domain (bounded by the size of the telescope pupil). \bigskip

In the following, we describe the annular telescope aperture mask by $\Omega=\Omega_y\times \Omega_x \subseteq \left[-D/2,D/2\right]^2$ as shown in Figure~\ref{fig:omega}. Single lines of the annular aperture are represented by $\Omega_x = \left[a_x,b_x\right]$ and $\Omega_y = \left[a_y,b_y\right]$, with $a_x <b_x$, $a_y<b_y$ being the borders of the pupil for fixed $x$ and $y$ correspondingly. The sensor provides measurement on the region of the CCD-detector $D$. Throughout the paper, we do not distinguish between pupil and CCD-detector and assume $D=\Omega$. Further, the limitation onto the CCD detector (indicated as multiplication with a characteristic function of $D$) is not marked explicitly for the underlying operators for simplicity of notation. However, please keep in mind the compact support of the considered functions because of the restricted size of the aperture and the CCD-detector. For instance, we consider the norms in $\mathcal{L}_2\left(\R^2\right)$ but since the operators map from and to functions with compact support on $\Omega$ it is equivalent to considering the norms in $\mathcal{L}_2\left(\Omega\right)$. 

\section {Pyramid and roof wavefront sensor forward operators}\label{chap:pwfs_forward_models}

The pyramid WFS sensor model is non-linear and extensive~\cite{KoVe07,ISTRE1}. The measurements $(s^n_x, s^n_y)$ of the non-modulated sensor are connected to the wavefront $\Phi$ via a combination of $1$d and $2$d Hilbert transforms of $\sin{\Phi}, \cos\Phi$ and their multiplications. Circular modulation complicates the model even more which makes the task of inverting the operator $\boldsymbol P$, i.e., the full Fourier optics based model of the sensor rather difficult. \bigskip \par
In the following $\boldsymbol P^{\{n,c\}} = [\boldsymbol P_x^{\{n,c\}},\boldsymbol P_y^{\{n,c\}}]$ denote the operators, that describe the pyramid sensor forward model, $\boldsymbol R^{\{n,c,l\}} = [\boldsymbol R_x^{\{n,c,l\}}, \boldsymbol R_y^{\{n,c,l\}}]$ the roof wavefront sensor operators. The linearized roof sensor operators are indicated by $\boldsymbol R^{\{n,c,l\},lin}=[\boldsymbol R_x^{\{n,c,l\},lin},\boldsymbol R_y^{\{n,c,l\},lin}]$. The superscripts $\{n,c,l\}$ represent the regime in which the sensor is operated -- no modulation, circular or linear modulation applied. \bigskip

%
%
%

Atmospheric turbulence is usually introduced by either using the von Karman or Kolmogorov turbulence model which describe the energy distribution of turbulent motions \cite{Roddier,RoWe96}. Following the Kolmogorov model, considerations in~\cite{Ellerbroek02} initiate to expect smoother functions for representing atmospheric turbulence than just $\mathcal{L}_2$ with a high probability. This leads to the assumption of $\Phi \in \mathcal{H}^{11/6}\left(\R^2\right)$ as already suggested in~\cite{Eslitz13}. \bigskip

We start with the derivation of the analytic pyramid wavefront sensor model without taking interference effects between the four pupils on the CCD-detector into account. This pyramid sensor model is known as the transmission mask model. For the modulated sensor, we define a modulation parameter
\begin{equation} \label{eq:mod_param}
 \alpha_{\lambda} = \dfrac{2\pi \alpha} {\lambda} ,
\end{equation}
with $\alpha = r\lambda/D$ for a positive integer $r$ representing the modulation radius and $\lambda$ the sensing wavelength. 

\subsection{Pyramid forward model without interference (transmission mask model)}

For the pyramid wavefront sensor, we only consider the non-modulated and the sensor with circular modulation since they make sense from the physical point of view.
\begin{definition}\label{3.1}
We introduce the operators $\boldsymbol P_x^{\{n,c\}}$ in $x$-direction given by
\begin{align}\label{eq:3.2}
\left(\boldsymbol P_x^{\{n,c\}}\Phi\right)(x,y)&:=\dfrac{1}{\pi}\ \mathcal{X}_{\Omega}(x,y) \ p.v. \ \int\limits_{\Omega_y}{\dfrac{\sin{\left[\Phi(x',y)-\Phi(x,y)\right] \cdot k^{\{n,c\}} (x'-x)  }}{x'-x}\ dx'} \\ \notag
&+ \dfrac{1}{\pi^3}\ \mathcal{X}_{\Omega_y}(x) \ p.v. \ \int\limits_{\Omega_y}{\int\limits_{\Omega_x}{\int\limits_{\Omega_x}{\dfrac{\sin{\left[\Phi(x',y')-\Phi(x,y'')\right] \cdot l^{\{n,c\}} (x'-x,y''-y') }}{(x'-x)(y'-y)(y''-y)} \ dy'' \ }dy' \  } dx' }
\end{align}
and $\boldsymbol P_y^{\{n,c\}}$ in $y$-direction given by
\begin{align}\label{eq:3.3}
\left(\boldsymbol P_y^{\{n,c\}}\Phi\right)(x,y)&:=\dfrac{1}{\pi}\ \mathcal{X}_{\Omega}(x,y) \ p.v. \ \int\limits_{\Omega_x}{\dfrac{\sin{\left[\Phi(x,y')-\Phi(x,y)\right] \cdot k^{\{n,c\}} (y'-y) }}{y'-y}\ dy'} \\ \notag
&+ \dfrac{1}{\pi^3}\ \mathcal{X}_{\Omega_x}(y) \ p.v. \ \int\limits_{\Omega_y}{\int\limits_{\Omega_x}{\int\limits_{\Omega_y}{\dfrac{\sin{\left[\Phi(x',y')-\Phi(x'',y)\right] \cdot l^{\{n,c\}} (x''-x',y'-y) }}{(x'-x)(y'-y)(x''-x)} \ dx'' \ }dy' \ } dx' }.
\end{align}
The functions $k^{\{n,c\}}$ are defined by $k^n (x) := 1 , \ k^c (x): = J_0 (\alpha_{\lambda} x) $,
and the functions $l^{\{n,c\}}$ by $l^n(x,y): = 1$ and
\begin{equation*}
 l^c(x,y) := \dfrac{1}{T} \int\limits_{-T/2}^{T/2} \cos [ \alpha_{\lambda} x \sin (2\pi t/T) ] \cos [ \alpha_{\lambda} y \cos (2\pi t/T) ] \ dt .
\end{equation*}
The function $J_0$ denotes the zero-order Bessel function of the first kind given by
\begin{equation*}
 J_0 (x) = \dfrac{1}{\pi} \int\limits_0^{\pi} \cos ( x \sin t ) \ dt = \dfrac{1}{\pi} \int\limits_0^{\pi} \cos ( x \cos t ) \  dt 
\end{equation*}
with the modulation parameter $\alpha_{\lambda}$ defined in~\eqref{eq:mod_param}.
\end{definition}

Note that the kernels of the involved integral operators are strongly singular. Hence, they are defined in the p.v. (principal value) meaning, i.e., the integrals above are meant in the sense of
$$p.v. \ \int\limits_{\Omega_y}\dfrac{\Phi(x',y)}{x'-x}\ dx' = \lim\limits_{\delta\rightarrow 0^+}\int\limits_{\Omega_y\backslash\left[x-\delta,x+\delta\right]}\dfrac{\Phi(x',y)}{x'-x}\ dx'$$
for any wavefront $\Phi$. In the following, $p.v.  \int\limits_{\Omega_y}\int\limits_{\Omega_x}\int\limits_{\Omega_y}$ is always meant as abbreviation of $p.v. \int\limits_{\Omega_y}p.v.  \int\limits_{\Omega_x}p.v.  \int\limits_{\Omega_y}$ in the context of the pyramid sensor operator.


\begin{theorem}\label{theor_meas}
The relation between the pyramid wavefront sensor data and the incoming wavefront is given by
\begin{equation}\label{eq:3.4a}
\begin{split}
s_x^{\{n,c\}}(x,y) &= -\tfrac{1}{2}\left(\boldsymbol P_x^{\{n,c\}}\Phi\right)(x,y), \\ 
s_y^{\{n,c\}}(x,y) & = \tfrac{1}{2}\left(\boldsymbol P_y^{\{n,c\}}\Phi\right)(x,y),
\end{split}
\end{equation}
where $\boldsymbol P^{\{n,c\}}$ denote the pyramid sensor operators defined in~\eqref{eq:3.2}~-~\eqref{eq:3.3}.
\end{theorem}

\begin{proof}
The proof is given in the Appendix.
\end{proof}


\subsection{Roof sensor forward models (transmission mask model)}

The roof WFS constitutes a part of the pyramid WFS. In the roof sensor, the pyramidal prism is replaced by two orthogonally placed two-sided roof prisms, resulting in a decoupling of $x$- and $y$-direction. Therefore, the roof WFS operators $\boldsymbol R^{\{n,c\}}=[\boldsymbol R_x^{\{n,c\}},\boldsymbol R_y^{\{n,c\}}]$ are much simpler, as it contains only (variations of) 1d Hilbert transforms in one particular direction~\cite{KoVe07,ISTRE2,Veri04}.  \\ Due to the physical setup of the roof WFS, linear modulation induces characteristics which are of interest especially for roof wavefront sensors. Hence, we additionally investigate the linear modulated roof wavefront sensor model. In case of circular modulation the amplitude of modulation is assumed to be $\alpha_\lambda$, already introduced in~\eqref{eq:mod_param}. For linear modulation, $2\alpha$ denotes the angle of ray displacement along the desired direction, which is equivalent to the assumed circular modulation. We substitute the zero-order Bessel function by a $\sinc$-term in order to describe a linear modulation instead of a circular one. \\ For further investigations, we again do not take interference between the four beams into account and consider the roof WFS operator $\boldsymbol R^{\{n,c,l\}}$ on the one hand as a standalone wavefront sensor and on the other hand as an approximation to the pyramid WFS $\boldsymbol P^{\{n,c\}}$.

\begin{definition}\label{eq:3.5}
We introduce the operators $\boldsymbol R_x^{\{n,c,l\}}$ and  $\boldsymbol R_y^{\{n,c,l\}}$ by
\begin{align}\label{eq:3.6a}
\left( \boldsymbol R_x^{\{n,c,l\}} \Phi\right) (x,y) &:=\mathcal{X}_{\Omega}(x,y)   \dfrac { 1 } { \pi }\ p.v. \int \limits_{\Omega_y} \dfrac{ \sin [ \Phi (x',y) - \Phi (x,y) ] \cdot k^{\{n,c,l\}} (x'-x) } {x'-x} \ dx', \\ \label{eq:3.6b}
\left( \boldsymbol R_y^{\{n,c,l\}} \Phi\right) (x,y) &:=\mathcal{X}_{\Omega}(x,y)   \dfrac { 1 } { \pi }\ p.v. \int \limits_{\Omega_x} \dfrac{ \sin [ \Phi (x,y') - \Phi (x,y) ] \cdot k^{\{n,c,l\}}(y'-y) } {y'-y} \ dy' 
\end{align}
with modulation functions $k^n (x) := 1$, $k^c (x) := J_0 (\alpha_{\lambda} x)$, and $k^l (x) := \sinc (\alpha_{\lambda} x)$.
\end{definition}

\begin{theorem}
Using the operators defined in~\eqref{eq:3.6a}~-~\eqref{eq:3.6b} the measurements of the roof WFS in the transmission mask model are written as
\begin{align*}
s_x^{\{n,c,l\}} (x,y) &= -\tfrac{1}{2}\left( \boldsymbol R_x^{\{n,c,l\}} \Phi \right) (x,y), \\
s_y^{\{n,c,l\}} (x,y) &= \tfrac{1}{2}\left( \boldsymbol R_y^{\{n,c,l\}} \Phi \right) (x,y). 
\end{align*}
\end{theorem}

\begin{proof}
See~\cite{BuDa06,Shatokhina_PhDThesis,Veri04}.
\end{proof}

The operators $\boldsymbol P^{\{n,c\}}_x$ and $\boldsymbol P^{\{n,c\}}_y$ as well as $\boldsymbol R^{\{n,c,l\}}_x$ and $\boldsymbol R^{\{n,c,l\}}_y$ are constructed in the same way, one only has to interchange the roles of $x$ and $y$ in the model. In the following, we will concentrate on the operators $\boldsymbol P^{\{n,c\}}_x$ and $\boldsymbol R^{\{n,c,l\}}_x$ since the obtained results can easily be transferred to the operators $\boldsymbol P^{\{n,c\}}_y$ and $\boldsymbol R^{\{n,c,l\}}_y$ as well. Let us now analyze the pyramid and roof sensor operators in more detail. 

\begin{prop}\label{boundedness_roof}
The non-linear operators $\boldsymbol R^{\{n,c,l\}}:\mathcal{H}^{11/6}\left(\R^2\right)\rightarrow\mathcal{L}_2\left(\R^2\right)$, representing roof wavefront sensors, are well-defined operators between the above given spaces.
\end{prop}
\begin{proof}
From the proof of Theorem \ref{theor_meas} it follows that the pyramid sensor operators and (further) the roof sensor operators are non-linear and well-defined for any wavefront $\Phi \in \mathcal{H}^{11/6}\left(\R^2\right)$. It remains to show $\left(\boldsymbol R^{\{n,c,l\}}\Phi\right) \in \mathcal{L}_2\left(\R^2\right)$.
The proof uses the boundedness$\textsuperscript{(1)}$ $\left|\sin\left(\Phi\right)\right|\le \left| \Phi \right|$ and the H\"{o}lder continuity$\textsuperscript{(2)}$ with $\alpha=5/6$ and H\"{o}lder constant $C>0$ in one direction of any function $\Phi \in \mathcal{H}^{11/6}\left(\R^2\right)$ (cf Sobolev embedding theorem, e.g.,~\cite[Theorem~5.4]{Adams75}). We start with showing that the integrand of~\eqref{eq:3.6a} in fact is integrable for $\mathcal{D}\left(\boldsymbol R^{\{n,c,l\}}\right) \subseteq \mathcal{H}^{11/6}\left(\R^2\right)$. Together with $\left|k^{\{n,c,l\}}\right| \le 1$, this infers from
\begin{align*}
 \int \limits_{\Omega_y} \left|\dfrac{ \sin [ \Phi (x',y) - \Phi (x,y) ] \cdot k^{\{n,c,l\}} (x'-x) } {x'-x}\right| \ dx' 
&\overset{(1)}{\le} \int \limits_{\Omega_y} \dfrac{\left|  \Phi (x',y) - \Phi (x,y) \right| } {\left|x'-x\right|} \ dx' \\
&\overset{(2)}{\le} C \int \limits_{\Omega_y} \dfrac{\left|  x' - x \right|^{5/6} } {\left|x'-x\right|} \ dx'
= C \int \limits_{\Omega_y} \dfrac{1} {\left|x'-x\right|^{1/6}} \ dx'  \\
&= \dfrac{6C}{5}\left(\left(b_y-x\right)^{5/6}+\left(x+a_y\right)^{5/6}\right) < \ \infty
\end{align*}
for $x \in \Omega_y$, $y \in \Omega_x$. As the proper integral exists, the Cauchy principal value exists as well. It follows that the $p.v.$ meaning is negligible in~\eqref{eq:3.6a}~-~\eqref{eq:3.6b} for $\mathcal{D}\left(\boldsymbol R^{\{n,c,l\}}\right) \subseteq \mathcal{H}^{11/6}\left(\R^2\right)$. By usage of the Cauchy-Schwarz inequality~(C-S), we obtain that the $\mathcal{L}_2$-norm
\begin{align*}
\left|\left|\boldsymbol R_x^{\{n,c,l\}}\Phi\right|\right|_{\mathcal{L}_2\left(\R^2\right)}^2 
&= \int\limits_{\R^2}\left|\boldsymbol R_x^{\{n,c,l\}}\Phi\left(x,y\right)\right|^2 \ d\left(x,y\right) \\
&= \int\limits_{\R^2}\left|\mathcal{X}_{\Omega}(x,y)   \dfrac { 1 } { \pi } \int \limits_{\Omega_y} \dfrac{ \sin [ \Phi (x',y) - \Phi (x,y) ] \cdot k^{\{n,c,l\}} (x'-x) } {x'-x} \ dx'\right|^2 \ d\left(x,y\right) \\
&\overset{\mathrm{C-S}}{\le} \dfrac { 1 } { \pi^2 } \int\limits_{\Omega}\left( \int \limits_{\Omega_y} \left|\dfrac{ \sin [ \Phi (x',y) - \Phi (x,y) ] \cdot k^{\{n,c,l\}} (x'-x) } {x'-x}  \right|^2 dx'\cdot \int\limits_{\Omega_y} 1\ dx' \right)d\left(x,y\right)\\
&\le \dfrac { \left|\Omega_y\right| } { \pi^2 } \int\limits_{\Omega} \int \limits_{\Omega_y} \left(\dfrac{\left| \sin [ \Phi (x',y) - \Phi (x,y) ]\right| } {\left|x'-x\right|} \right)^2\ dx' \ d\left(x,y\right) \\
&\overset{(1)}{\le} \dfrac {  \left|\Omega_y\right|  } { \pi^2 } \int\limits_{\Omega}\int \limits_{\Omega_y} \left(\dfrac{\left|  \Phi (x',y) - \Phi (x,y) \right| } {\left|x'-x\right|} \right)^2 \ dx' \ d\left(x,y\right) \\
&\overset{(2)}{\le} \dfrac { C^2  \left|\Omega_y\right| } { \pi^2 } \int\limits_{\Omega} \int \limits_{\Omega_y}\left( \dfrac{\left|  x' - x \right|^{5/6} } {\left|x'-x\right|} \right)^2 dx' \ d\left(x,y\right) \\
&= \dfrac { C^2  \left|\Omega_y\right|  } { \pi^2 } \int\limits_{\Omega}\int \limits_{\Omega_y}\dfrac{1} {\left|x'-x\right|^{1/3}}  \ dx' \ d\left(x,y\right)  \\
&\le\sup\limits_{x\in\Omega_y}\left|M\left(x\right)\right| \dfrac {C^2 \left|\Omega\right| \left|\Omega_y\right|  } {\pi^2 }    < \ \infty
\end{align*}
is finite as for $\Omega_y\subseteq \left[-D/2,D/2\right]$ holds
\begin{equation}\label{integrate_M}
M(x,y):=\int \limits_{\Omega_y}\dfrac{1} {\left|x'-x\right|^{1/3}}  \ dx'  \le \int \limits_{-D/2}^{D/2}\dfrac{1} {\left|x'-x\right|^{1/3}}  \ dx' = \dfrac{3}{2} \left[\left(\dfrac{D}{2}-x\right)^{2/3}+\left(x+\dfrac{D}{2}\right)^{2/3}\right] =:M(x)
\end{equation}
and further
\begin{equation}\label{integrate_M1}
\sup\limits_{x\in\Omega_y}\left|M\left(x\right)\right| < \infty.
\end{equation}
Hence, the roof sensor operators maps incoming wavefronts in $\mathcal{H}^{11/6}\left(\R^2\right)$ to measurements in $\mathcal{L}_2\left(\R^2\right)$ with compact support on the telescope pupil/detector.
\end{proof}

From the considerations in the above proof it follows the following statement:
\begin{remark}
If $\Phi \in \mathcal{H}^{11/6}\left(\R^2\right)$, the principal value integrals of the operators~\eqref{eq:3.6a}~-~\eqref{eq:3.6b} describing the roof wavefront sensor model coincide with the standard definition of Lebesgue integrals.
\end{remark}

\begin{prop}\label{boundedness}
The non-linear operators $\boldsymbol P^{\{n,c\}}:\mathcal{H}^{11/6}\left(\R^2\right)\rightarrow\mathcal{L}_2\left(\R^2\right)$, representing pyramid wavefront sensors, are well-defined operators between the above given spaces.
\end{prop}
\begin{proof}
As already shown in the proof of Theorem~\ref{theor_meas}, the pyramid sensor operators are non-linear, well-defined operators for any wavefront $\Phi \in \mathcal{H}^{11/6}\left(\R^2\right)$. In order to verify that $\left(\boldsymbol P^{\{n,c\}}\Phi\right) \in \mathcal{L}_2$ we split the corresponding operators into two parts:
\begin{equation*}\label{pyr_2parts1}
\left(\boldsymbol P_x^{\{n,c\}}\Phi\right)(x,y) = \left(\boldsymbol R_x^{\{n,c\}}\Phi\right)(x,y) + \left(\boldsymbol S_x^{\{n,c\}}\Phi\right)(x,y)
\end{equation*}
with the roof sensor operators $\boldsymbol R_x^{\{n,c\}}$ defined in~\eqref{eq:3.6a} and the second term
$$\left(\boldsymbol S_x^{\{n,c\}}\Phi\right)(x,y) :=  \dfrac{1}{\pi^3}\ \mathcal{X}_{\Omega_y}(x) \ p.v. \ \int\limits_{\Omega_y}{\int\limits_{\Omega_x}{\int\limits_{\Omega_x}{\dfrac{\sin{\left[\Phi(x',y')-\Phi(x,y'')\right] \cdot l^{\{n,c\}} (x'-x,y''-y') }}{(x'-x)(y'-y)(y''-y)} \ dy'' \ }dy' \ } dx' }.$$
With Proposition~\ref{boundedness_roof} and $$\left|\left|\boldsymbol P_x^{\{n,c\}}\Phi\right|\right|_{\mathcal{L}_2\left(\R^2\right)}  \le \left|\left|\boldsymbol R_x^{\{n,c\}}\Phi\right|\right|_{\mathcal{L}_2\left(\R^2\right)} +\left|\left|\boldsymbol S_x^{\{n,c\}}\Phi\right|\right|_{\mathcal{L}_2\left(\R^2\right)}, $$ it remains to show $\left|\left|\boldsymbol S_x^{\{n,c\}}\Phi\right|\right|_{\mathcal{L}_2\left(\R^2\right)} < \infty$. \bigskip

First, we focus on the non-modulated sensor and consider the operator $\boldsymbol S_x^n$.

The proof uses the $\mathcal{L}_p$-boundedness of the classical Hilbert transform for $1<p<\infty$ as found in, e.g.,~\cite{Butz71}. For our purposes, we define the Hilbert transforms $\boldsymbol H_x$ in $x$- direction and $\boldsymbol H_y$ in $y$-direction by
\begin{equation}\label{hilbert_def}
\boldsymbol H_x\Phi(x,y) :=\dfrac{1}{\pi} \ p.v. \ \int_{-\infty}^\infty\dfrac{\Phi(x',y)}{x'-x} \ dx', \qquad \qquad \qquad \boldsymbol H_y\Phi(x,y) :=\dfrac{1}{\pi} \ p.v.  \int_{-\infty}^\infty\dfrac{\Phi(x,y')}{y'-y} \ dy' .
\end{equation}
\begin{theorem}[Theorem 8.1.12, \cite{Butz71}] \label{bound}
For $\Phi \in \mathcal{L}_p\left(\R^2\right)$, $1<p<\infty$, the Hilbert transform defined in~\eqref{hilbert_def} exists almost everywhere, belongs to $\mathcal{L}_p\left(\R^2\right)$ and satisfies
\begin{equation}\label{eq:bound}
\left|\left|\boldsymbol H_x\Phi\right|\right|_{\mathcal{L}_p}\le c_p\left|\left|\Phi\right|\right|_{\mathcal{L}_p} \qquad \qquad \text{and} \qquad \qquad \left|\left|\boldsymbol H_y\Phi\right|\right|_{\mathcal{L}_p}\le d_p\left|\left|\Phi\right|\right|_{\mathcal{L}_p}
\end{equation}
with constants $c_p,d_p>0$. \eqref{eq:bound} is often referred to as the \textit{Marcel Riesz inequality for Hilbert transforms}.
\end{theorem}
The 2d Hilbert transform $\boldsymbol H_{xy}$ is given as an operator $\boldsymbol H_{xy}:\mathcal{L}_p\left(\R^2\right)\rightarrow  \mathcal{L}_p\left(\R^2\right)$ for $1 < p < \infty$ as well since it can be considered as the composition $\boldsymbol H_{xy} = \boldsymbol H_{x} \circ \boldsymbol H_{y}$.

Using trigonometric formulas, we rewrite $\boldsymbol S_x^n$ into
{\small\begin{align*}
 \left(\boldsymbol S_x^n\Phi\right)(x,y) 
 &=  \dfrac{1}{\pi^3}\ \mathcal{X}_{\Omega_x}(y) \ p.v. \ \int\limits_{\Omega_y}{\int\limits_{\Omega_x}{\int\limits_{\Omega_y}{\dfrac{\sin{\left[\Phi(x',y')-\Phi(x,y'')\right]}}{(x'-x)(y'-y)(y''-y)} \ dy'' \ }dy' \ } dx' } \\
 &= \dfrac{1}{\pi^3}\ \mathcal{X}_{\Omega_x}(y) \ p.v. \ \int\limits_{\Omega_y}{\int\limits_{\Omega_x}{\int\limits_{\Omega_y}{\dfrac{\sin\left[\Phi(x',y')\right]\cos\left[\Phi(x,y'')\right]-\cos\left[\Phi(x',y')\right]\sin\left[\Phi(x,y'')\right]}{(x'-x)(y'-y)(y''-y)} \ dy'' \ }dy' \ } dx' } \\
 &=\left(\boldsymbol H_{xy}\left(\sin\Phi\right)\right)\left(x,y\right) \cdot \left(\boldsymbol H_y\left(\cos\Phi\right)\right)\left(x,y\right) - \left(\boldsymbol H_{xy}\left(\cos\Phi\right)\right)\left(x,y\right) \cdot \left(\boldsymbol H_y\left(\sin\Phi\right)\right)\left(x,y\right)
\end{align*}
}and obtain for $2\le p<\infty$ 
\begin{align*}
\pi^3\left|\left|\boldsymbol S_x^n\Phi\right|\right|_{\mathcal{L}_{p/2}}
&=\left|\left|\left(\boldsymbol H_{xy}\left(\sin\Phi\right)\right) \cdot \left(\boldsymbol H_y\left(\cos\Phi\right)\right)- \left(\boldsymbol H_{xy}\left(\cos\Phi\right)\right) \cdot \left(\boldsymbol H_y\left(\sin\Phi\right)\right)\right|\right|_{\mathcal{L}_{p/2}} \\
&\le \left|\left|\left(\boldsymbol H_{xy}\left(\sin\Phi\right)\right) \cdot \left(\boldsymbol H_y\left(\cos\Phi\right)\right)\right|\right|_{\mathcal{L}_{p/2}} + \left|\left| \left(\boldsymbol H_{xy}\left(\cos\Phi\right)\right) \cdot \left(\boldsymbol H_y\left(\sin\Phi\right)\right)\right|\right|_{\mathcal{L}_{p/2}} \\
&\le \left|\left|\left(\boldsymbol H_{xy}\left(\sin\Phi\right)\right)\right|\right|_{\mathcal{L}_{p}} \left|\left|\left(\boldsymbol H_y\left(\cos\Phi\right)\right)\right|\right|_{\mathcal{L}_{p}} + \left|\left| \left(\boldsymbol H_{xy}\left(\cos\Phi\right)\right)\right|\right|_{\mathcal{L}_{p}} \left|\left| \left(\boldsymbol H_y\left(\sin\Phi\right)\right)\right|\right|_{\mathcal{L}_{p}} \\
&< \infty 
\end{align*} 
by using the generalized H\"{o}lder inequality. The multiplication with the characteristic functions is omitted above due to simplicity of notation.

It follows that
\begin{equation*}\label{L_p_bound}
\left(\boldsymbol P^n\Phi\right) \in \mathcal{L}_p\left(\R^2\right)
\end{equation*}
 for $1\le p < \infty$ and $\Phi \in \mathcal{H}^{11/6}\left(\R^2\right)$. \bigskip

In order to prove $\left(\boldsymbol P^c\Phi\right) \in \mathcal{L}_2$ we use relation~\eqref{eq:mod_sx_general} of the proof of Theorem~\ref{theor_meas} between non-modulated and modulated pyramid data, i.e.,
\begin{equation}\label{eq:rel_nomod_mod} 
 s_x^{c} (x,y) = \dfrac{1}{T} \int\limits_{-T/2}^{T/2} s_x^n (x,y,t)\  dt .
\end{equation}
Let T denote one full time period and $t\in\left[-T/2,T/2\right]$. For deriving the time-dependent non-modulated pyramid sensor data, we introduce an operator $\boldsymbol M_t^{mod}$ given by
\begin{equation*}
\left(\boldsymbol M_t^{mod}\Phi\right)\left(x,y\right):=\Phi\left(x,y\right)+\Phi^{mod}\left(x,y,t\right)
\end{equation*}
for the periodic tilt $ \Phi^{mod}$ inducing modulation. As in~\eqref{eq:periodic_tilt}, this tilt is represented by
\begin{equation*}
 \Phi^{mod} (x,y,t) = \alpha_{\lambda} ( x \sin (2\pi t/T) + y \cos (2\pi t/T) ). 
\end{equation*} 
Due to the structure of $\Phi^{mod}$ and its compact support on the aperture, it holds $\Phi^{mod}\left(\cdot,\cdot,t\right) \in \mathcal{H}^{11/6}\left(\R^2\right)$ for all $t \in \left[-T/2,T/2\right]$. This gives a continuous map $\boldsymbol M^{mod}_t: \mathcal{H}^{11/6}\left(\R^2\right) \rightarrow  \mathcal{H}^{11/6}\left(\R^2\right)$ and further
\begin{equation}\label{eq:time_bound}
\left(\boldsymbol P^n\boldsymbol M^{mod}_t\Phi\right) \in \mathcal{L}_p\left(\R^2\right) \qquad \qquad \forall t\in\left[-T/2,T/2\right], 1\le p < \infty.
\end{equation}
Using $s_x^n\left(\cdot,\cdot,t\right)=\left(\boldsymbol P^n\boldsymbol M^{mod}_t\Phi\right)$, \eqref{eq:rel_nomod_mod}, and the generalized Minkowski's integral inequality$\textsuperscript{(1)}$ (cf, e.g.,~\cite[Theorem 202]{hardy88},\cite{stein70}), we obtain
\begin{align*}
\left|\left|\boldsymbol s_x^c\right|\right|_{\mathcal{L}_{p}}
&=\left( \int\limits_{\R^2}\left|s_x^{c}\left(x,y\right)\right|^p \ d\left(x,y\right)\right)^{1/p} \\
&=\left(\int\limits_{\R^2}\left| \dfrac{1}{T} \int\limits_{-T/2}^{T/2} \left(\boldsymbol P^n\boldsymbol M_t^{mod}\Phi\right)\left(x,y\right)\  dt \ \right|^p d\left(x,y\right)\right)^{1/p} \\
&\overset{(1)}{\le} \dfrac{1}{T} \int\limits_{-T/2}^{T/2}\left(\ \int\limits_{\R^2} \left|\left(\boldsymbol P^n\boldsymbol M_t^{mod}\Phi\right)\left(x,y\right)\right|^p \ d\left(x,y\right)\right)^{1/p} \ dt \\
&= \dfrac{1}{T} \int\limits_{-T/2}^{T/2} \left|\left|\boldsymbol P^n\boldsymbol M_t^{mod}\Phi\right|\right|_{\mathcal{L}_p} \ dt
\overset{\eqref{eq:time_bound}}{<} \infty,
\end{align*}
which shows $\boldsymbol P^c\Phi \in \mathcal{L}_p\left(\R^2\right),1\le p<\infty$ for the modulated pyramid sensor operators. Note that we omitted the factor $\left(-\tfrac{1}{2}\right)$ from~\eqref{eq:3.4a}.
\end{proof}
Merely the light which is captured on the telescope pupil $\Omega$ influences the pyramid sensor response. We consider only the light falling on the aperture $\left(\mathcal{X}_{\Omega} \cdot \  \Phi\right)$ but use the notation $\Phi \in \mathcal{H}^{11/6}\left(\R^2\right)$ for both considered variants of wavefronts with and without compact support on the telescope pupil.

Alternatively, one can also define the pyramid sensor operators as $\boldsymbol P_x^{\{n,c\}}:\mathcal{H}^{11/6}\left(\Omega\right)\rightarrow\mathcal{L}_2\left(\Omega\right)$, i.e., on the telescope aperture, and $\boldsymbol P_y^{\{n,c\}}, \boldsymbol R_x^{\{n,c,l\}}, \boldsymbol R_y^{\{n,c,l\}}$ respectively. However, since we apply, e.g., Plancherel's theorem we define the operators on the whole $\R^2$ but keep in mind that one can always restrict to the region of the telescope aperture and CCD-detector.  
More precisely, the pyramid sensor operators are applied to functions with compact support on the pupil $\Omega$ and map to functions with compact support on the detector $\Omega$.

\subsection{Linearization of the roof sensor operators}
Linear approximations of WFS operators around the zero phase are sufficient in closed loop AO in which the wavefront sensor measures already corrected and very small incoming wavefronts.
The linearization of the operators can be obtained by different ways, e.g., by replacing
$$ \sin [ \Phi (x',y) - \Phi (x,y) ] \approx \Phi (x',y) - \Phi (x,y)  $$
which is valid in case of  
$$\left|\Phi (x',y) - \Phi (x,y)\right| << 1 .$$
Linearizations $\boldsymbol R^{\{n,c,l\},lin}$ based on these approximations were already considered in~\cite{BuDa06,KoVe07,Shat13,Veri04}.

We concentrate on linear approximations for the roof wavefront sensor operators by means of the Fr\'{e}chet derivative. We start by calculating the G\^{a}teaux derivatives. Then, we show that the G\^{a}teaux derivatives coincide with the Fr\'{e}chet derivatives and finally, we evaluate the corresponding linearizations.

\begin{theorem}\label{3.9_gateaux}
The G\^{a}teaux derivatives $\left(\boldsymbol R_x^{\{n,c,l\}}\right)'(\Phi) \in \mathcal{L}\left(\mathcal{H}^{11/6},\mathcal{L}_2\right)$ at $\Phi \in \mathcal{D}\left(\boldsymbol R_x^{\{n,c,l\}}\right)$ of the non-linear roof sensor operators $\boldsymbol R_x^{\{n,c,l\}}:\mathcal{D}\left(\boldsymbol R_x^{\{n,c,l\}}\right)\subseteq \mathcal{H}^{11/6}\left(\R^2\right)\rightarrow \mathcal{L}_2\left(\R^2\right)$ defined in~\eqref{eq:3.6a} are given by
{\small\begin{equation}\label{theo_gateaux_def}
\left(\left(\boldsymbol R_x^{\{n,c,l\}}\right)'(\Phi)\ \psi\right)\left(x,y\right) =\mathcal{X}_\Omega\left(x,y\right)\dfrac{1}{\pi}  \int\limits_{\Omega_y}{\dfrac{\cos{\left[\Phi(x',y)-\Phi(x,y)\right]}\left[\psi(x',y)-\psi(x,y)\right]\cdot k^{\{n,c,l\}}\left(x'-x\right)}{(x'-x)}\ dx'}
\end{equation}
}and $\left(\boldsymbol R_y^{\{n,c,l\}}\right)'(\Phi) \in \mathcal{L}\left(\mathcal{H}^{11/6},\mathcal{L}_2\right)$ respectively.
\end{theorem}
\begin{proof}
We utilize the representation (cf Taylor's theorem with Lagrange form of the remainder)
\begin{eqnarray}\label{eq:taylorsin}
\sin\left(\Phi+\psi\right) &= \sin\left(\Phi\right)+\sin'\left(\Phi\right)\psi+\dfrac{1}{2}\sin''\left(\Phi+\theta\psi\right) \psi^2
\end{eqnarray}
for a $\theta = \theta\left(\Phi,\psi\right) \in \left(0,1\right)$.  For simplicity of notation, we omit the multiplication with the characteristic function of the aperture in front of every integral and the multiplication with the modulation kernels $k^{\{n,c,l\}}$ as these functions are independent of the phase $\Phi$ anyway.

The G\^{a}teaux derivatives $\left(\boldsymbol R_x^{\{n,c,l\}}\right)'(\Phi)$ are computed as
{\small \begin{equation*} \label{gateaux_deriv}
\begin{split}
\left(\left(\boldsymbol R_x^{\{n,c,l\}}\right)'(\Phi)\ \psi\right)\left(x,y\right)
&=\lim\limits_{t \rightarrow 0} \dfrac{\left(\boldsymbol R_x^{\{n,c,l\}}\left(\Phi+t\psi\right)\right)\left(x,y\right)-\left(\boldsymbol R_x^{\{n,c,l\}}\Phi\right)\left(x,y\right)}{t} \\
&=\lim\limits_{t\rightarrow 0} \dfrac{1}{\pi} \int\limits_{\Omega_y}{\dfrac{\sin {\left[\Phi(x',y)+t\psi(x',y)-\Phi(x,y)-t\psi(x,y)\right]}-\sin{\left[\Phi(x',y)-\Phi(x,y)\right]}}{t(x'-x)}\ dx'} \\ 
&\overset{\eqref{eq:taylorsin}}{=}\lim\limits_{t\rightarrow 0}  \dfrac{1}{\pi}  \int\limits_{\Omega_y}\dfrac{\sin'{\left[\Phi(x',y)-\Phi(x,y)\right]}\left[t\psi(x',y)-t\psi(x,y)\right]}{t(x'-x)} \\ 
&+\dfrac{\dfrac{1}{2}\sin''{\left[\Phi(x',y)-\Phi(x,y)+\theta t \left[\psi(x',y)-\psi(x,y)\right]\right]}t^2\left[\psi(x',y)-\psi(x,y)\right]^2}{t(x'-x)}\ dx' \\
&=\dfrac{1}{\pi}  \int\limits_{\Omega_y}{\dfrac{\cos{\left[\Phi(x',y)-\Phi(x,y)\right]}\left[\psi(x',y)-\psi(x,y)\right]}{(x'-x)}\ dx'}.
\end{split}
\end{equation*}
}Obviously, the G\^{a}teaux derivatives are linear in $\psi$. 

For any $\Phi \in \mathcal{H}^{11/6}\left(\R^2\right)$ we deduce that the G\^{a}teaux derivatives are bounded, and further continuous in the direction $\psi$, i.e., $\left(\boldsymbol R_x^{\{n,c,l\}}\right)'(\Phi) \in \mathcal{L}\left(\mathcal{H}^{11/6},\mathcal{L}_2\right)$, by showing
\begin{equation*}
\begin{split}
\left|\left|\left(\boldsymbol R_x^{\{n,c,l\}}\right)'(\Phi)\right|\right|_{\mathcal{L}\left(\mathcal{H}^{11/6},\mathcal{L}_2\right)} 
&= \sup\limits_{\left|\left|\psi\right|\right|_{\mathcal{H}^{11/6}}=1} \left|\left|\left(\left(\boldsymbol R_x^{\{n,c,l\}}\right)'(\Phi)\ \psi\right)\right|\right|_{\mathcal{L}_2\left(\R^2\right)} <\infty.
\end{split}
\end{equation*} 
By application of the Cauchy-Schwarz inequality, $\left|k^{\{n,c,l\}}\right| \le 1$, as well as the H\"{o}lder continuity$\textsuperscript{(1)}$ with $\alpha=5/6$ and H\"{o}lder constant $C>0$ of any function in $\mathcal{H}^{11/6}\left(\R^2\right)$, the statement results from
{\footnotesize \begin{equation*}
\begin{split}
\left|\left|\left(\left(\boldsymbol R_x^{\{n,c,l\}}\right)'(\Phi)\ \psi\right)\right|\right|^2_{\mathcal{L}_2\left(\R^2\right)} 
&= \int\limits_{\R^2} \left|\left(\left(\boldsymbol R_x^{\{n,c,l\}}\right)'(\Phi)\ \psi\right)\left(x,y\right)\right|^2 d\left(x,y\right) \\
&= \int\limits_{\Omega} \left|\dfrac{1}{\pi}\ \int\limits_{\Omega_y} \dfrac{\cos\left[\Phi\left(x',y\right)-\Phi\left(x,y\right)\right]\left[\psi\left(x',y\right)-\psi\left(x,y\right)\right]\cdot k^{\{n,c,l\}}\left(x'-x\right)}{x'-x} \ dx' \right|^2 d\left(x,y\right) \\
&\overset{\mathrm{C-S}}{\le}\dfrac{1}{\pi^2}\int\limits_{\Omega}\left(\int\limits_{\Omega_y}\left|\cos\left[\Phi\left(x',y\right)-\Phi\left(x,y\right)\right]\cdot k^{\{n,c,l\}}\left(x'-x\right)\right|^2 dx' \right)\\
&\cdot \left( \int\limits_{\Omega_y}\left|\dfrac{\psi\left(x',y\right)-\psi\left(x,y\right)}{x'-x}\right|^2 \ dx' \right) d\left(x,y\right) \\
&\overset{(1)}{\le} \dfrac{1}{\pi^2}\int\limits_{\Omega}\left(\int\limits_{\Omega_y}1 \ dx' \right)\cdot \left( C^2 \int\limits_{\Omega_y}\left(\dfrac{\left|x'-x\right|^{5/6}}{\left|x'-x\right|}\right)^2 \ dx' \right) d\left(x,y\right)  \\
&=\dfrac{C^2\left|\Omega_y\right|}{\pi^2}\int\limits_{\Omega} \int\limits_{\Omega_y}\dfrac{1}{\left|x'-x\right|^{1/3}} \ dx' \ d\left(x,y\right)  \overset{\eqref{integrate_M}-\eqref{integrate_M1}}{<} \infty.
\end{split}
\end{equation*} 
}
\end{proof}

\begin{theorem}
The G\^{a}teaux derivatives~\eqref{theo_gateaux_def} coincide with the Fr\'{e}chet derivatives.
\end{theorem}
\begin{proof}
For the proof we use the following assertion stated in, e.g.,~\cite{bakushinsky_2018,werner2007funktionalanalysis}.  
\begin{prop}[Theorem III.5.4, \cite{werner2007funktionalanalysis}; p.10, \cite{bakushinsky_2018}]
If the G\^{a}teaux derivatives $\left(\boldsymbol R_x^{\{n,c,l\}}\right)'(\Phi)$ exist for all $\Phi$ from a neighborhood of $\Phi_0 \in \mathcal{D}\left(\boldsymbol R_x^{\{n,c,l\}}\right)$ and the mappings $\Phi \rightarrow \left(\boldsymbol R_x^{\{n,c,l\}}\right)'(\Phi)$ are continuous from $\mathcal{H}^{11/6}\left(\R^2\right)$ into $\mathcal{L}\left(\mathcal{H}^{11/6},\mathcal{L}_2\right)$ at $\Phi = \Phi_0$, then $\boldsymbol R^{\{n,c,l\}}_x$ are Fr\'{e}chet differentiable at~$\Phi_0$.
\end{prop}
Hence, it suffices to show that for any $\Phi_1,\Phi_2 \in \mathcal{H}^{11/6}\left(\R^2\right)$ holds
\begin{equation*}\label{contin_frechet}
\begin{split}
\Big|\Big|\left(\boldsymbol R_x^{\{n,c,l\}}\right)'(\Phi_1)-\left(\boldsymbol R_x^{\{n,c,l\}}\right)'(\Phi_2)\Big|\Big|^2_{\mathcal{L}\left(\mathcal{H}^{11/6},\mathcal{L}_2\right)} 
\le \tilde{C} \left|\left|\Phi_1-\Phi_2\right|\right|_{\mathcal{H}^{11/6}}^2,
\end{split}
\end{equation*} 
i.e.,
\begin{equation*}
\begin{split}
\sup\limits_{\left|\left|\psi\right|\right|_{\mathcal{H}^{11/6}}=1} \Big|\Big|\underbrace{\left(\left(\boldsymbol R_x^{\{n,c,l\}}\right)'(\Phi_1)\ \psi\right)-\left(\left(\boldsymbol R_x^{\{n,c,l\}}\right)'(\Phi_2)\ \psi\right)}_{:=U_1}\Big|\Big|^2_{\mathcal{L}_2} 
\le \tilde{C} \left|\left|\Phi_1-\Phi_2\right|\right|_{\mathcal{H}^{11/6}}^2
\end{split}
\end{equation*} 
with $\tilde{C} < \infty$. 

\begingroup
\allowdisplaybreaks

Under the Lipschitz continuity$\textsuperscript{(2)}$ of the cosine function with the Lipschitz constant $L>0$ we obtain
{\small \begin{align*}
 \left|\left|U_1\right|\right|^2_{\mathcal{L}_2} 
\overset{{\color{white}3.2611}}{=}{}& \int\limits_{\R^2} \left|\left(\left(\boldsymbol R_x^{\{n,c,l\}}\right)'(\Phi_1)\ \psi\right)\left(x,y\right)-\left(\left(\boldsymbol R_x^{\{n,c,l\}}\right)'(\Phi_2)\ \psi\right)\left(x,y\right)\right|^2 d\left(x,y\right) \\
\overset{{\color{white}3.2611}}{=}{}& \int\limits_{\Omega} \left|\dfrac{1}{\pi}\int\limits_{\Omega_y}\Big[\cos\left[\Phi_1\left(x',y\right)-\Phi_1\left(x,y\right)\right]-\cos\left[\Phi_2\left(x',y\right)-\Phi_2\left(x,y \right)\right]\Big] \right.\\
&\left.\cdot \dfrac{\left[\psi\left(x',y\right)-\psi\left(x,y\right)\right]\cdot k^{\{n,c,l\}}\left(x'-x\right)}{x'-x} \ dx' \right|^2 d\left(x,y\right) \\
\overset{{\color{white}3i}\mathrm{C-S}{\color{white}31}}{\le}{}& \dfrac{1}{\pi^2} \int\limits_{\Omega}\left(\ \int\limits_{\Omega_y}\Big|\cos\left[\Phi_1\left(x',y\right)-\Phi_1\left(x,y\right)\right]-\cos\left[\Phi_2\left(x',y\right)-\Phi_2\left(x,y \right)\right]\Big|^2 dx'\right) \\ 
&\cdot 
\left(\ \int\limits_{\Omega_y}\left| \dfrac{\left[\psi\left(x',y\right)-\psi\left(x,y\right)\right]\cdot k^{\{n,c,l\}}\left(x'-x\right)}{x'-x} \right|^2 dx' \right) d\left(x,y\right) \\
\overset{(1),(2)}{\le}{}&  \dfrac{L^2C^2}{\pi^2} \int\limits_{\Omega}\left(\ \int\limits_{\Omega_y}\left|\Phi_1\left(x',y\right)-\Phi_1\left(x,y\right)-\Phi_2\left(x',y\right)+\Phi_2\left(x,y \right)\right|^2 dx'\right) \\ 
&\cdot 
\left(\ \int\limits_{\Omega_y}\left( \dfrac{\left|x'-x\right|^{5/6}}{\left|x'-x\right|} \right)^2 dx' \right) d\left(x,y\right) \\
\overset{{\color{white}3.2611}}{=}{}&  \dfrac{L^2C^2}{\pi^2} \int\limits_{\Omega}\left( \ \int\limits_{\Omega_y}\left|\left[\Phi_1\left(x',y\right)-\Phi_2\left(x',y\right)\right]-\left[\Phi_1\left(x,y\right)-\Phi_2\left(x,y \right)\right]\right|^2 dx'\right) \\ 
&\cdot 
\left(\ \int\limits_{\Omega_y} \dfrac{1}{\left|x'-x\right|^{1/3}} \ dx' \right) d\left(x,y\right) \\
\overset{\eqref{integrate_M}}{\le}{}&  \dfrac{2L^2C^2}{\pi^2} \int\limits_{\Omega}\left(\ \int\limits_{\Omega_y}\left|\Phi_1\left(x',y\right)-\Phi_2\left(x',y\right)\right|^2+\left|\Phi_1\left(x,y\right)-\Phi_2\left(x,y \right)\right|^2 dx'\right) \\
&\cdot M(x,y) \ d\left(x,y\right) \\
\overset{{\color{white}3.2611}}{\le}{}&  \sup\limits_{x\in\Omega_y}\left|M(x)\right| \dfrac{2L^2C^2}{\pi^2} \left(\ \int\limits_{\Omega_y}\int\limits_{\Omega_x} \int\limits_{\Omega_y}\left|\Phi_1\left(x',y\right)-\Phi_2\left(x',y\right)\right|^2 dx' \ dy \ dx  \right. \\
&+ \left.\int\limits_{\Omega_y}\int\limits_{\Omega_x} \int\limits_{\Omega_y} \left|\Phi_1\left(x,y\right)-\Phi_2\left(x,y \right)\right|^2 dx \ dy \ dx'  \right)\\
\overset{{\color{white}3.2611}}{=}{}&  \sup\limits_{x\in\Omega_y}\left|M(x)\right| \dfrac{2L^2C^2\left|\Omega_y\right|}{\pi^2} \left(\left|\left|\Phi_1-\Phi_2\right|\right|_{\mathcal{L}_2}^2 + \left|\left|\Phi_1-\Phi_2\right|\right|_{\mathcal{L}_2}^2 \right)\\
\overset{{\color{white}3.2611}}{=}{}&  \sup\limits_{x\in\Omega_y}\left|M(x)\right| \dfrac{4L^2C^2\left|\Omega_y\right|}{\pi^2} \left|\left|\Phi_1-\Phi_2\right|\right|_{\mathcal{L}_2}^2 \\
\overset{{\color{white}3.2611}}{\le}{}&  \sup\limits_{x\in\Omega_y}\left|M(x)\right| \dfrac{4L^2C^2\left|\Omega_y\right|}{\pi^2} \left|\left|\Phi_1-\Phi_2\right|\right|_{\mathcal{H}^{11/6}}^2
\overset{\eqref{integrate_M1}}{\le} \tilde{C} \left|\left|\Phi_1-\Phi_2\right|\right|_{\mathcal{H}^{11/6}}^2,
\end{align*}
}i.e., the mapping $\Phi \rightarrow \left(\boldsymbol R_x^{\{n,c,l\}}\right)'(\Phi)$ is continuous. Thus, the operators $\boldsymbol R_x^{\{n,c,l\}}$ representing the roof sensor are Fr\'{e}chet differentiable. \endgroup \end{proof}
\begin{theorem}\label{3.9}
The linearizations $\boldsymbol R_x^{\{n,c,l\},lin}:\mathcal{H}^{11/6}\left(\R^2\right)\rightarrow\mathcal{L}_2\left(\R^2\right)$ by means of the Fr\'{e}chet derivative of the operators $\boldsymbol R^{\{n,c,l\}}$ introduced in~\eqref{eq:3.6a}-\eqref{eq:3.6b} are given by
{\small\be{3.81}
\left(\boldsymbol R_x^{\{n,c,l\},lin}\Phi\right)(x,y):= \left( \left( \boldsymbol R_x^{\{n,c,l\}} \right)'  (0) \ \Phi\right)\left(x,y\right)=\mathcal{X}_{\Omega}(x,y)\dfrac{1}{\pi} \int\limits_{\Omega_y}{\dfrac{ [\Phi(x',y)-\Phi(x,y)] \cdot k^{\{n,c,l\}} (x'-x) }{x'-x}\ dx'}
\ee
}for $x$-direction and $\boldsymbol R^{\{n,c,l\},lin}_y$ accordingly.
\end{theorem}

\begin{proof}
The claim immediately follows from $\left(\boldsymbol R_x^{\{n,c,l\},lin}\Phi\right):= \left( \left( \boldsymbol R_x^{\{n,c,l\}} \right)'  (0) \ \Phi\right)$, i.e., considering the Fr\'{e}chet derivatives \eqref{theo_gateaux_def} at $\Phi=0$ and in direction $\Phi$.
\end{proof}

\subsection{Simplified linearized operators $\boldsymbol L_x^{\{n,c,l\}}$}

Variations of the finite Hilbert transform operator allow to simplify the linear approximations of the roof sensor model. Let us, thus, consider the following operators.

\begin{definition}
We define the integral operators $\boldsymbol L^{\{n,c,l\}} = \left[\boldsymbol L_x^{\{n,c,l\}},\boldsymbol L_y^{\{n,c,l\}}\right]$ by
\begin{equation}\label{3b.13}
( \boldsymbol L_x^{\{n,c,l\}} \Phi ) (x,y) := \dfrac { 1 } { \pi } \ p.v. \int \limits_{\Omega_y} \dfrac{  \Phi (x',y) k^{\{n,c,l\}} \left(x'-x\right)} {x'-x} \ dx' 
\end{equation}
and $\boldsymbol L_y^{\{n,c,l\}}$ accordingly.
\end{definition}
As derived in \cite{fefferman81,stein70}, $\boldsymbol L^{\{n,c,l\}}$ are bounded operators on $\mathcal{L}_p\left(\R^2\right)$ for $1 < p < \infty$ due to the structure of the functions $k^{\{n,c,l\}}$ introducing modulation. According to the pyramid and roof sensor model, we consider the operators as $\boldsymbol L^{\{n,c,l\}}: \mathcal{H}^{11/6}\left(\\R^2\right)\rightarrow\mathcal{L}_2\left(\R^2\right)$. In case of no modulation the operator $\boldsymbol L^n_x$ coincides with the finite Hilbert transform -- a singular Cauchy integral operator. Using the above defined operators, the linearized roof sensor measurements read as
\begin{eqnarray} \label{3.11}
s_x^{\{n,c,l\},lin} (x,y) &=& -\tfrac{1}{2}\left( \boldsymbol R^{\{n,c,l\},lin}_x\Phi \right) (x,y) \notag \\ 
&=& -\tfrac{1}{2}\mathcal{X}_{\Omega}(x,y)\left[\left( \boldsymbol L_x^{\{n,c,l\}} \Phi \right) (x,y) - \Phi (x,y) \cdot \left(\boldsymbol L_x^{\{n,c,l\}} 1\right) (x,y)\right].
\end{eqnarray}

Equation \eqref{3.11} provides two possibilities for wavefront reconstruction. As it is shown in \cite{ISTRE2,ISTRE1}, when an AO system enters the closed loop, the first term in the forward model gains in importance and an assumption of neglecting the second term in the reconstruction procedure is justifiable. Therefore, one could either use the full linearized roof sensor model or ignore the second term  $\left(\boldsymbol L_x^{\{n,c,l\}} 1\right)$ as it is done in several existing algorithms for pyramid wavefront sensor \cite{ShatHut_spie2018_overview}, e.g., the Preprocessed Cumulative Reconstructor with Domain decomposition (P-CuReD) \cite{Shat_SPIE,Shat13}, the Pyramid Fourier Transform Reconstructor \cite{Shat17_ao4elt5_clif,Shat17}, the Finite Hilbert Transform Reconstructor \cite{Shatokhina_PhDThesis}, or the Singular Value Type Reconstructor \cite{Hut17}.

\section {Adjoint operators}\label{chap:adjoint_operators}

Several iterative algorithms for solving inverse problems (e.g., the Landweber iteration, steepest descent method or conjugate gradient method for the normal equation) include the application of adjoint operators. In order to make these methods suitable for wavefront reconstruction, we derive the adjoints of the underlying operators.

First, we will evaluate the Fourier transforms of the one-term assumptions $\boldsymbol L^{\{n,c,l\}}$ defined in \eqref{3b.13} and afterwards use Plancherel's theorem to calculate the corresponding adjoints. 

The underlying operators $\boldsymbol L^{\{n,c,l\}}$ are defined from $\mathcal{H}^{11/6}\left(\R^2\right)$ into $\mathcal{L}_2\left(\R^2\right)$. In order to calculate the corresponding adjoint operators from $\mathcal{L}_2\left(\R^2\right)$ into $\mathcal{H}^{11/6}\left(\R^2\right)$, we introduce the embedding operator 
\begin{equation}\label{eq:embed_op}
i_s: \mathcal{H}^{11/6} \rightarrow \mathcal{L}_2
\end{equation}
and derive the adjoints according to \cite{RaTesch04} as, e.g., $$\left( \boldsymbol L^{\{n,c,l\},lin}_x\right)^*:\mathcal{L}_2\rightarrow \mathcal{H}^{11/6} \qquad \text{with} \qquad \left( \boldsymbol L^{\{n,c,l\},lin}_x\right)^* = i_s^*\left( \tilde{\boldsymbol L}^{\{n,c,l\},lin}_x\right)^*$$ for $\left(\tilde{\boldsymbol L}^{\{n,c,l\},lin}_x\right)^*: \mathcal{L}_2 \rightarrow \mathcal{L}_2$. For simplicity, we use the notation $\left(\boldsymbol L^{\{n,c,l\},lin}_x\right)^*$ instead of $\left(\tilde{\boldsymbol L}^{\{n,c,l\},lin}_x\right)^*$ and omit the multiplication with the aperture mask in the following. The adjoints of the roof sensor operators are considered accordingly. \bigskip

\begin{prop} \label{f3.13}
The 1d Fourier transforms in $x$-direction of the operators $\boldsymbol L_x^{\{n,c,l\}}$ defined in \eqref{3b.13} are given by
\begin{equation}\label{fourier_mult}
\left( \boldsymbol L^{\{n,c,l\}}_x\Phi\right)^{\widehat{\phantom{x}}}(\xi,y) =c^{\{n,c,l\}}\left(\xi\right)\cdot\widehat{\Phi}\left(\xi,y\right) 
\end{equation}
with
\begin{equation}\label{ft_1}
\begin{split}
c^n\left(\xi\right) =i\sgn\left(\xi\right) \\ 
\end{split}
\end{equation}
for the non-modulated sensor,
\begin{equation}\label{ft_2}
\begin{split}
c^c\left(\xi\right) =i\begin{cases} \sgn\left(\xi\right), \qquad \qquad &\mathrm{for} \ \left|\xi\right| > \tfrac{\alpha}{\lambda}, \\ \tfrac{2}{\pi}\arcsin\left(\xi\tfrac{\lambda}{\alpha}\right), & \mathrm{for} \ \left|\xi\right| \le \tfrac{\alpha}{\lambda} \end{cases}
\end{split}
\end{equation}
for the circularly modulated sensor, and
\begin{equation} \label{ft_3}
\begin{split}
c^l\left(\xi\right) =i\begin{cases} \sgn\left(\xi\right), \qquad \qquad &\mathrm{for} \ \left|\xi\right| > \tfrac{\alpha}{\lambda}, \\  \xi\tfrac{\lambda}{\alpha}, & \mathrm{for} \ \left|\xi\right| \le \tfrac{\alpha}{\lambda} \end{cases}
\end{split}
\end{equation}
for the linearly modulated sensor. The Fourier transforms of $\boldsymbol L_y^{\{n,c,l\}}$ are represented accordingly.
\end{prop}
\begin{proof}
Similar considerations are contained in~\cite{Shat13,Veri04}. Since we examine the 1d Hilbert transform $\mathcal{L}_2\left(\R\right)\rightarrow \mathcal{L}_2\left(\R\right)$ we fix $y \in \Omega_x$ and investigate the operators $\boldsymbol L_x^{\{n,c,l\}}:\mathcal{H}^{11/6}\left(\R\right)\subseteq\mathcal{L}_{2}\left(\R\right)\rightarrow \mathcal{L}_2\left(\R\right)$ defined according to~\eqref{3b.13} without indicating the fixed $y$ specifically. For fixed $x \in \Omega_y$ the operators $\boldsymbol L_y^{\{n,c,l\}}:\mathcal{H}^{11/6}\left(\R\right)\subseteq\mathcal{L}_{2}\left(\R\right)\rightarrow \mathcal{L}_2\left(\R\right)$ are analyzed respectively. 
We introduce the even kernel functions $v^{\{n,c,l\}}$ by
$$v^{\{n,c,l\}}(x):=p.v. \dfrac{k^{\{n,c,l\}}(x)}{x}$$ and obtain
\begin{equation}\label{conv_meas}
\begin{split}
\left(\boldsymbol L_x^{\{n,c,l\}}\Phi\right)\left(x,y\right) ={}& -\dfrac{1}{\pi} \left(\Phi\left(\cdot,y\right)\ast v^{\{n,c,l\}}\right)\left(x\right)  \\ 
={}&- \dfrac{1}{\pi} \int\limits_{-\infty}^{\infty}\Phi\left(x',y\right)v^{\{n,c,l\}}\left(x-x'\right) \ dx' \\
={}& \lim\limits_{\delta\rightarrow 0^+}\dfrac{1}{\pi} \left[\int\limits_{-\infty}^{x-\delta}\dfrac{\Phi\left(x',y\right)k^{\{n,c,l\}}\left(x'-x\right)}{x'-x} \ dx'  \right.\\
&+\left. \int\limits_{x+\delta}^{\infty}\dfrac{\Phi\left(x',y\right)k^{\{n,c,l\}}\left(x'-x\right)}{x'-x} \ dx'\right].
\end{split}
\end{equation}
By the convolution theorem, the 1d convolution in~\eqref{conv_meas} is a multiplication in the Fourier domain, i.e.,
\begin{equation}\label{conv_op}
\left(  \boldsymbol L^{\{n,c,l\}}_x\Phi\right)^{\widehat{\phantom{x}}}\left(\xi,y\right) = - \sqrt{\dfrac{2}{\pi}} \ \widehat{\Phi}\left(\xi,y\right) \cdot \widehat{v^{\{n,c,l\}}}\left(\xi\right).
\end{equation}
As already used in~\cite{Shat13,Veri04}, the Fourier transforms of the kernel functions are calculated as
\begin{equation*}
\begin{split}
\widehat{v^n}\left(\xi\right) =-i\sqrt{\tfrac{\pi}{2}}\sgn\left(\xi\right) \\ 
\end{split}
\end{equation*}
for the non-modulated sensor,
\begin{equation*}
\begin{split}
\widehat{v^c}\left(\xi\right) =-i\begin{cases} \sqrt{\tfrac{\pi}{2}}\sgn\left(\xi\right), \qquad \qquad &\mathrm{for} \ \left|\xi\right| > \tfrac{\alpha}{\lambda}, \\  \sqrt{\tfrac{2}{\pi}}\arcsin\left(\xi\tfrac{\lambda}{\alpha}\right), & \mathrm{for} \ \left|\xi\right| \le \tfrac{\alpha}{\lambda} \end{cases}
\end{split}
\end{equation*}
for the circularly modulated sensor, and
\begin{equation*}
\begin{split}
\widehat{v^l}\left(\xi\right) =-i\begin{cases} \sqrt{\tfrac{\pi}{2}}\sgn\left(\xi\right), \qquad \qquad &\mathrm{for} \ \left|\xi\right| > \tfrac{\alpha}{\lambda}, \\  \sqrt{\tfrac{\pi}{2}}\xi\tfrac{\lambda}{\alpha}, & \mathrm{for} \ \left|\xi\right| \le \tfrac{\alpha}{\lambda} \end{cases}
\end{split}
\end{equation*}
for the linearly modulated sensor.

The claim of the Proposition follows by~\eqref{conv_op} and $c^{\{n,c,l\}} =-\sqrt{\tfrac{2}{\pi}}\widehat{v^{\{n,c,l\}}}$.
\end{proof}

Note that by the isometry of the Fourier transform we obtain $$\left|\left|\boldsymbol L_x^{\{n,c,l\}}\Phi\left(\cdot,y\right)\right|\right|_{\mathcal{L}_2}=\left|\left|\left(\boldsymbol L_x^{\{n,c,l\}}\Phi\right)^{\widehat{\phantom{x}}}\left(\cdot,y\right)\right|\right|_{\mathcal{L}_2}=\left|\left|c^{\{n,c,l\}} \cdot \widehat{\Phi}\left(\cdot,y\right)\right|\right|_{\mathcal{L}_2}=\tilde{c}\left|\left|\widehat{\Phi}\left(\cdot,y\right)\right|\right|_{\mathcal{L}_2},$$
i.e., 
\begin{equation*} \label{boundedness_L}
\left|\left|\boldsymbol L_x^{\{n,c,l\}}\Phi\left(\cdot,y\right)\right|\right|_{\mathcal{L}_2} = \tilde{c}\left|\left|\Phi\left(\cdot,y\right)\right|\right|_{\mathcal{L}_2}
\end{equation*}
for a constant $0<\tilde{c}<\infty$.


\subsection {Adjoint operators $\left( \boldsymbol L^{\{n,c,l\}}_x\right)^*$ in $\mathcal{L}_2\left(\R^2\right)$}

As previously mentioned, it is sufficient to derive the adjoints as operators from $\mathcal{L}_2$ into itself and use the embedding operator \eqref{eq:embed_op} in order to obtain adjoint operators from $\mathcal{L}_2$ into $ \mathcal{H}^{11/6}$ \cite{RaTesch04}.

\begin{prop}\label{3.13}
The adjoints $\left( \boldsymbol L^{\{n,c,l\}}_x\right)^*:\mathcal{L}_2\left(\R^2\right)\rightarrow \mathcal{L}_2\left(\R^2\right)$ of the operators $\boldsymbol L_x^{\{n,c,l\}}$ defined in \eqref{3b.13} are given by
\be{3.14} \notag
\left( \left( \boldsymbol L^{\{n,c,l\}}_x\right)^*\Psi\right)(x,y) = -\dfrac{1}{\pi}\ p.v.\int\limits_{\Omega_y}{\dfrac{\Psi(x',y) \cdot k^{\{n,c,l\}} (x'-x) }{x'-x}\ dx'},
\ee
and $\left( \boldsymbol L_y^{\{n,c,l\}}\right)^*$ accordingly, i.e., $ \boldsymbol L^{\{n,c,l\}}$ are skew-adjoint in $\mathcal{L}_2\left(\R^2\right)$.
\end{prop}

\begin{proof}
The proof is performed in the Fourier domain. We have to consider the $\mathcal{L}_2\left(\C\right)$-inner product due to $c^{\{n,c,l\}} \in \C$ defined in~\eqref{ft_1}~-~\eqref{ft_3}. We use Plancherel's theorem and the equality $c^{\{n,c,l\}}=-\overline{c^{\{n,c,l\}}}$ for the complex conjugate$\textsuperscript{(1)}$.

For any $\Phi,\Psi \in \mathcal{L}_2\left(\R^2\right)$ with support on $\Omega$ and $y \in \Omega_x$ holds
\begin{align*}
\langle \left(\boldsymbol L^{\{n,c,l\}}_x\Phi\right)\left(\cdot,y\right),\Psi\left(\cdot,y\right)\rangle_{\mathcal{L}_2\left(\C\right)} 
&= \langle \left(\boldsymbol L^{\{n,c,l\}}_x\Phi\right)^{\widehat{\phantom{x}}}\left(\cdot,y\right),\widehat{\Psi}\left(\cdot,y\right)\rangle_{\mathcal{L}_2\left(\C\right)}  \\
&= \int\limits_{\R}{{\left(\boldsymbol L^{\{n,c,l\}}_x\Phi\right)^{\widehat{\phantom{x}}}(\xi,y)\overline{\widehat{\Psi}(\xi,y)} \ d\xi }} \\
&\overset{\eqref{fourier_mult}}{=} \int\limits_{\R}{{\left(c^{\{n,c,l\}}\left(\xi\right) \cdot \widehat{\Phi}(\xi,y)\right)\overline{\widehat{\Psi}(\xi,y)} \ d\xi }} \\
&\overset{(1)}{=}- \int\limits_{\R}{{ \widehat{\Phi}(\xi,y)\overline{\left(c^{\{n,c,l\}}\left(\xi\right)\cdot\widehat{\Psi}(\xi,y)\right)} \ d\xi}} \\
&=\langle \widehat{\Phi}\left(\cdot,y\right),\left(\left(\boldsymbol L^{\{n,c,l\}}_x\right)^* \Psi\right)^{\widehat{\phantom{x}}}\left(\cdot,y\right)\rangle_{\mathcal{L}_2\left(\C\right)} \\
&=\langle \Phi\left(\cdot,y\right),\left(\left( \boldsymbol L^{\{n,c,l\}}_x\right)^* \Psi\right)\left(\cdot,y\right)\rangle_{\mathcal{L}_2\left(\C\right)}
 \end{align*}
 with $\left(\left(\boldsymbol L^{\{n,c,l\}}_x\right)^*\Psi\right)^{\widehat{\phantom{x}}}\left(\xi,y\right) = - c^{\{n,c,l\}}\left(\xi\right) \widehat{\Psi}\left(\xi,y\right)$, i.e., $\left( \boldsymbol L^{\{n,c,l\}}_x\right)^* = -\boldsymbol L^{\{n,c,l\}}_x$.
\end{proof}

\subsection {Adjoint operators $\left( \boldsymbol R^{\{n,c,l\},lin}_x\right)^*$ in $\mathcal{L}_2\left(\R^2\right)$}

\begin{prop}\label{3.15}
The adjoints $\left(\boldsymbol R_x^{\{n,c,l\},lin}\right)^*:\mathcal{L}_2\left(\R^2\right)\rightarrow \mathcal{L}_2\left(\R^2\right)$ of the linearized roof sensor operators $\boldsymbol R_x^{\{n,c,l\},lin}$ defined in \eqref{eq:3.81} are given by
\be{3.16} 
\begin{split}\notag
\left(\left(\boldsymbol R_x^{\{n,c,l\},lin}\right)^*\Psi\right)(x,y) &= \left(\left( \boldsymbol L^{\{n,c,l\}}_x\right)^*\Psi \right) (x,y)  - \Psi(x,y) \left( \boldsymbol L^{\{n,c,l\}}_x 1 \right) (x,y) \\
&=-\dfrac{1}{\pi} \ p.v. \int\limits_{\Omega_y}{\dfrac{ [\Psi(x',y)+\Psi(x,y)] \cdot k^{\{n,c,l\}} (x'-x) }{x'-x}\ dx'}
\end{split}
\ee
and $\left(\boldsymbol R_y^{\{n,c,l\},lin}\right)^*$ respectively.
\end{prop}

\begin{proof} We choose any $\Phi,\Psi \in \mathcal{L}_2\left(\R^2\right)$ with support on the telescope pupil $\Omega$. Due to the linearity of the inner product and with representation \eqref{3.11}, it holds that
\begin{align*}\label{eq:3.17}
\langle \left(\boldsymbol R^{\{n,c,l\},lin}_x\Phi\right),\Psi\rangle&= \langle \mathcal{X}_{\Omega}\left(\boldsymbol L^{\{n,c,l\}}_x\Phi-\Phi \boldsymbol L^{\{n,c,l\}}_x 1\right),\Psi\rangle \\ \notag
&= \langle \boldsymbol L^{\{n,c,l\}}_x\Phi-\Phi \boldsymbol L^{\{n,c,l\}}_x 1,\Psi\rangle \\
& = \langle \boldsymbol L^{\{n,c,l\}}_x\Phi,\Psi\rangle - \langle \Phi \boldsymbol L^{\{n,c,l\}}_x 1,\Psi \rangle \notag\\
&=\langle \Phi,\left( \boldsymbol L^{\{n,c,l\}}_x\right)^*\Psi\rangle - \langle \Phi,\Psi \boldsymbol L^{\{n,c,l\}}_x 1\rangle \notag \\
&=\langle \Phi, \left( \boldsymbol L^{\{n,c,l\}}_x\right)^*\Psi-\Psi\boldsymbol L^{\{n,c,l\}}_x 1\rangle \notag \\
&=\langle \Phi,\left(\left(\boldsymbol R_x^{\{n,c,l\},lin}\right)^*\Psi\right)\rangle, \notag
\end{align*}
where we consider the inner product with respect to ${\mathcal{L}_2\left(\R^2\right)}$ or ${\mathcal{L}_2\left(\Omega\right)}$ respectively.
\end{proof}

\section*{Conclusion}

In this paper, we have considered the mathematical background for the pyramid wavefront sensor which is widely used in Adaptive Optics systems in areas such as astronomy, microscopy, and ophthalmology. The theoretical analysis of the forward operators of pyramid and roof wavefront sensors presented in this paper was aimed at the subsequent development of fast and stable algorithms for wavefront reconstruction from pyramid sensor data. The interference effects as well as the phase shift introduced by the pyramidal prism were neglected for simplicity. The analysis allows any kind of modulation (no, circular, linear) to be applied to the sensors. Due to the closed loop operation assumption, we can linearize the initially non-linear forward operators. Additionally, the closed loop setting allows to simplify the linearized operators further, ending up with measurements described by one term comparable to a finite Hilbert transform in case of the non-modulated sensor. Moreover, for the considered operators we derived the corresponding adjoint operators. These are used in the upcoming second part of the paper~\cite{HuSha18_2} which will be devoted to the application of iterative algorithms to the problem of wavefront reconstruction from pyramid wavefront sensor data. An extension of the analysis (linearization and calculation of adjoints) to the full pyramid sensor operator as well as the adaption of the aperture mask to segmented pupils on ELTs~\cite{engler_spie_2018,HuShaOb18,ObRafShaHu18_proc,schwartz_spie2018,schwartz_ao4elt5} is dedicated to future work.

\section*{Acknowledgements}

This work has been partly supported by the Austrian Federal Ministry of Science and Research (HRSM) and the Austrian Science Fund (F68-N36, project 5). The authors thank the Austrian Adaptive Optics team for fruitful discussions and especially Simon Hubmer for his support.

\section*{Appendix}\label{chap:appendix}

\begin{proof}[Proof of Theorem \ref{theor_meas}]
First, we consider the non-modulated case.\\
Let $\Phi \in \mathcal{H}^{11/6}\left(\R^2\right)$ denote the phase screen in radians coming into the telescope. Using wave optics based models and assuming constant amplitude over the full telescope pupil $\Omega=\Omega_y\times\Omega_x$, the complex amplitude (wave) $\psi_{aper}$ corresponding to this incoming phase reads as
\be{psi_aper}\notag
\psi_{aper}(x,y)=\mathcal{X}_{\Omega}(x,y) \cdot \exp\left(-i\Phi(x,y)\right).
\ee 
Note that $\psi_{aper} \in \mathcal{L}_p\left(\R^2\right)$ for all $1\le p \le \infty$ due to the 
compact support of $\psi_{aper}$, the continuity of $\Phi$, and further the continuity of $\psi_{aper}$ on $\Omega_y\times\Omega_x$. In order to assume the continuity of $\psi_{aper}$ on $\R^2$ we slightly modify the telescope aperture mask $\mathcal{X}_{\Omega}$ and approximate it by $\mathcal{X}^\epsilon_{\Omega} \in \mathcal{C}_0^\infty\left(\R^2\right)$ fulfilling $\mathcal{X}_{\Omega}=\mathcal{X}^\epsilon_{\Omega}$ on $\Omega$. 

The idea of the extended mask is to smoothen the sharp edges of the telescope pupil in order to guarantee $\psi_{aper}^\epsilon \in \mathcal{H}^{11/6}\left(\R^2\right)$ using
\be{psi_aper_epsilon} \notag
\psi_{aper}^\epsilon(x,y):=\mathcal{X}_{\Omega}^\epsilon(x,y) \cdot \exp\left(-i\Phi(x,y)\right).
\ee 
The above assertion is fulfilled for an approximation of the aperture mask denoted by $\mathcal{X}^\epsilon_{\Omega} \in \mathcal{C}^\infty_0\left(\R^2\right)$. One possible representation of the modification $\mathcal{X}^\epsilon_{\Omega}$ can be constructed utilizing the following Lemma for the sets $\Omega = \Omega_y\times\Omega_x = \left[a_y,b_y\right]\times\left[a_x,b_x\right] $ and $ \Omega^\epsilon = \left(a_y-\epsilon,b_y+\epsilon\right)
\times\left(a_x-\epsilon,b_x+\epsilon\right)$ for a small $\epsilon > 0$, i.e., $\Omega \subset \Omega^\epsilon$.

\begin{lemma}[Lemma 4.2, \cite{Triebel}]
Let $\Omega \subset \R^n$ and $\Omega^\epsilon \subset \R^n$ be bounded with $\overline{\Omega}\subset \Omega^\epsilon$ and $\Omega^\epsilon$ open. Then, there exists a real-valued function $\mathcal{X}^\epsilon_{\Omega} \in \mathcal{C}_0^\infty\left(\Omega^\epsilon\right)$ with $0\le \mathcal{X}^\epsilon_{\Omega}\left(z\right) \le 1$ for $z \in \Omega^\epsilon$ and $\mathcal{X}^\epsilon_{\Omega}\left(z\right)=1$ for $z \in \Omega$. 
\end{lemma}
Outside of $\Omega^\epsilon$ we extend $\mathcal{X}_{\Omega}^\epsilon$ with zeros and obtain $\mathcal{X}_{\Omega}^\epsilon \in \mathcal{C}_0^\infty\left(\R^2\right)$. \bigskip

We investigate the construction of the new smooth aperture mask $\mathcal{X}^\epsilon_{\Omega}$ in more detail. As in \cite{Triebel}, we consider the smooth function $f \in \mathcal{C}_0^\infty\left(\R^2\right)$,
\begin{align}\label{eq:smooth} \notag
f\left(z\right) := \begin{cases} c \ \exp\left({-\dfrac{1}{1-\left|z\right|^2}}\right),   &\text{for} \ \left|z\right| < 1, \\
0, &\text{for} \ \left|z\right| \ge 1,\end{cases} 
\end{align}
having compact support on $\left[-1,1\right]^2$. The constant $c  \in \R$ is chosen such that $$\int\limits_{\R^2} f(z) \ dz = \int \limits_{\left|z\right|\le 1} f(z) \ dz = 1.$$
Additionally, we introduce $$f^\epsilon\left(z\right) := \dfrac{1}{\epsilon^2} f\left(\dfrac{z}{\epsilon}\right)$$ for $\epsilon >0$. 
Using this function, S. L. Sobolev established a method \cite{Sobolew} which is utilized for the smoothing of the characteristic function describing the telescope pupil. With the coordinate transformation \be{s.123}z-\epsilon z'=z'',\ee for the function $\mathcal{X}_{\Omega} \in \mathcal{L}_p\left(\R^2\right)$, $1\le p \le \infty$ we build the average function $\mathcal{X}^\epsilon_{\Omega} \in \mathcal{C}_0^\infty\left(\R^2\right)$ by
\begin{align}\label{av_func} \notag
\mathcal{X}_{\Omega}^\epsilon(z) &= \int\limits_{\R^2} \mathcal{X}_{\Omega}\left(z-\epsilon z'\right)f\left(z'\right) \ dz'  \\ \notag
&\overset{\eqref{eq:s.123}}{=} \int\limits_{\R^2} \dfrac{1}{\epsilon^2}\mathcal{X}_{\Omega}\left(z''\right)f\left(\dfrac{z-z''}{\epsilon}\right) \ dz''  \\ \notag
&=\int\limits_{\left|z-z''\right|\le \epsilon}f^{\epsilon}\left(z-z''\right)\mathcal{X}_{\Omega}\left(z''\right) \ dz'' .
\end{align}
Altogether, $\psi_{aper}^\epsilon$ is an element of $\mathcal{H}^{11/6}\left(\R^2\right)$ using that from $\Phi \in \mathcal{H}^{11/6}$ it follows $e^{-i\Phi} \in  \mathcal{H}^{11/6}$ as stated in~\cite{Brez01a,Brez01}.

The adaption of the aperture mask is necessary to guarantee a well-defined mathematical derivation of the pyramid wavefront sensor model. Please note that $$\left|\left|\mathcal{X}_{\Omega}-\mathcal{X}^\epsilon_{\Omega}\right|\right|_{\mathcal{L}_2\left(\R^2\right)} = \mathcal{O}\left(\epsilon^2\right),$$ and therefore $\mathcal{X}^\epsilon_{\Omega}$ is an arbitrarily good approximation of the aperture mask $\mathcal{X}_{\Omega}$.

The following physical argument supports the usage of a smoothed aperture mask $\Omega^{\epsilon}$ instead of the one with the sharp edges. Both masks have a compact support. Therefore, on the Fourier domain they are both represented with infinite spectra. Since $\epsilon<d$ (with $d$ denoting the telescope subaperture size) is small, the difference between the two masks, when looking in the  Fourier domain, appears only in the very high frequency components. Because the pyramidal prism is a physical device of finite size, it anyway cuts off high frequency components of the input. Additionally, the sensor brings spatial discretization in the model due to subaperture averaging. As a result, the spectra of the resulting sensor data contain frequencies only up to a given cut-off frequency defined via the subaperture size $d$ as $f_{cut}=1/(2d)$. Therefore, in practice there is no difference between the smoothed and the sharp aperture mask. In the following and throughout the paper we write $\mathcal{X}_{\Omega}$ but keep in mind that we always mean $\mathcal{X}_{\Omega}^\epsilon$ for a small $\epsilon > 0$ for the pyramid wavefront sensor model to be well-defined. \bigskip

As a next step, we consider the point spread function (PSF) of the glass pyramid.

The PSF is the inverse Fourier transform of the optical transfer function (OTF) of the pyramidal prism
\begin{equation}\label{psf_otf_fourier}
 PSF_{pyr} = {\cal F}^{-1}_{2d} \{ OTF_{pyr} \} . 
\end{equation}
Within the transmission mask approach \cite{KoVe07}, the OTF only takes splitting of the light into account and ignores the phase shifts introduced by the pyramid facets.  It is represented as a sum of 2d Heaviside functions
\begin{equation*}
OTF_{pyr}\left(\xi,\eta\right) = \sum\limits_{m=0}^1{\sum\limits_{n=0}^1{T^{mn}\left(\xi,\eta\right)}}
\end{equation*}
with
\begin{align*}\label{heaviside} 
T^{mn}\left(\xi,\eta\right) = H_{2d}\left[\left(-1\right)^m\xi,\left(-1\right)^n\eta\right] := \begin{cases} 1, \qquad &\text{if} \ \left(-1\right)^m\xi \ge 0, \quad \left(-1\right)^n\eta \ge 0, \\ 0,  &\text{otherwise.}\end{cases}
\end{align*}

The 2d Heaviside function is the product of two 1d Heaviside functions
\begin{equation*}
H_{2d} (\xi, \eta) = H_{1d} (\xi) \cdot H_{1d} (\eta)
\end{equation*}
and the 1d Heaviside function can be represented as
\begin{equation*}
H_{1d} (\xi) = \dfrac{1}{2} + \dfrac{1}{2} \cdot \texttt{sgn} (\xi).
\end{equation*}
Therefore
\begin{eqnarray*}
H_{2d} (\xi, \eta) 
&=& \dfrac{1}{4} \left[ 1 + \texttt{sgn} (\xi) + \texttt{sgn} (\eta) + \texttt{sgn} (\xi) \cdot \texttt{sgn} (\eta) \right], 
\end{eqnarray*}
which for $g_m\left(\xi\right):= \texttt{sgn}\left(\left(-1\right)^m\xi\right)$ results in
\begin{eqnarray*}
T^{mn} (\xi,\eta) 
&=& \dfrac{1}{4}\left[1+g_m\left(\xi\right)+g_n\left(\eta\right)+g_m\left(\xi\right)\cdot g_n\left(\eta\right)\right].
\end{eqnarray*}
Please note that in the above notation, $g_m$ is always meant as function of the first variable $\xi$ and $g_n$ as function of the second variable $\eta$ for 2d considerations. Furthermore, $\mathcal{F}_x$ will denote the 1d Fourier transform in the first variable and $\mathcal{F}_y$ the 1d Fourier transform in the second variable.
With \eqref{psf_otf_fourier} and due to the linearity of the Fourier transform, the PSF of the pyramid is represented as a sum of four PSFs
\begin{eqnarray*}
PSF_{pyr} &=& {\cal F}_{2d}^{-1} \Big\{\sum\limits_{m=0}^1{\sum\limits_{n=0}^1{T^{mn}}}\Big\}
= \sum\limits_{m=0}^1\sum\limits_{n=0}^1{\cal F}_{2d}^{-1} \{ T^{mn} \}  \\
&=& \sum\limits_{m=0}^1\sum\limits_{n=0}^1\underbrace{ {\cal F}_{2d}^{-1} \left\{ \dfrac{1}{4}\left[1+g_m+g_n+g_m\cdot g_n\right] \right\} }_{=:PSF_{pyr}^{mn}}.
\end{eqnarray*}
The OTF is a sum of products of functions depending either on $\xi$ or $\eta$. Therefore, the inverse 2d Fourier transform reduces to products of 1d inverse Fourier transforms. For $PSF_{pyr}^{mn}$, we obtain
\begin{eqnarray*}
PSF_{pyr}^{mn}\left(x,y\right)
&=& \dfrac{1}{4}  \left[  {\cal F}_{x}^{-1} \left\{ 1 \right\} (x) \cdot {\cal F}_{y}^{-1} \left\{  1 \right\} (y)  \right. \\ \nonumber
&+& {\cal F}_{x}^{-1} \left\{ g_m  \right\} (x)  \cdot  {\cal F}_{y}^{-1} \left\{  1 \right\} (y)  \\ \nonumber
&+& {\cal F}_{x}^{-1} \left\{  1 \right\} (x)   \cdot   {\cal F}_{y}^{-1} \left\{g_n \right\} (y)   \\ \nonumber
&+& \left. {\cal F}_{x}^{-1} \left\{ g_m  \right\} (x)  \cdot  {\cal F}_{y}^{-1} \left\{ g_n\right\} (y) \right] .  \nonumber  
\end{eqnarray*}


\begingroup
\allowdisplaybreaks

The Fourier transforms of the involved constant and signum functions do only exist in a distributional sense.
For test functions $\varphi$, we introduce the delta distribution $\delta$ defined as $\langle \delta,\varphi\rangle = \varphi(0)$. This application is well-defined for continuous functions, e.g., $\varphi \in \mathcal{H}^{1/2+\epsilon}\left(\R\right)$, and, on account of this, $\delta \in \mathcal{H}^{-1/2-\epsilon}\left(\R\right)$ for $\epsilon > 0$. The distribution $(p.v. \ \tfrac{1}{x})$ is defined via the Cauchy principal value by $$\Big\langle \left(p.v.\ \dfrac{1}{x}\right),\varphi\Big\rangle = \lim\limits_{\epsilon\rightarrow 0^+}\int\limits_{\left|x\right|>\epsilon}\dfrac{\varphi(x')}{x'} \ dx' = \pi \left(\boldsymbol H_x\varphi\right)\left(0\right)$$ for the 1d Hilbert transform defined according to \eqref{hilbert_def}. As the Hilbert transform $\boldsymbol H_x:\mathcal{L}_2\left(\R\right)\rightarrow\mathcal{L}_2\left(\R\right)$ is a well-defined operator (see, e.g., \cite[Theorem 8.1.12]{Butz71}), the evaluation of $\langle (p.v.\ \tfrac{1}{x}),\psi\rangle$ is well-defined for $\psi \in  \mathcal{H}^{1/2+\epsilon}\left(\R\right) \subset \mathcal{L}_2\left(\R\right)$, and thus $(p.v.\ \tfrac{1}{x}) \in \mathcal{H}^{-1/2-\epsilon}\left(\R\right)$. \\
Specifically, 
\begin{equation*} \label{eq:FT_1}
{\cal F}_{x}^{-1} \{ 1 \}  (x) = \sqrt{2\pi} \cdot \delta (x) \in \mathcal{H}^{-1/2-\epsilon}\left(\R\right), 
\end{equation*}
\begin{equation*} \label{eq:FT_sgn}
{\cal F}_{x}^{-1} \{ \texttt{sgn} (\cdot) \} (x) = i \cdot \sqrt{ \dfrac{2}{\pi} } \cdot p.v. \ \dfrac {1}{x} \in \mathcal{H}^{-1/2-\epsilon}\left(\R\right) ,
\end{equation*}
and in the same way
\begin{equation*}
{\cal F}_{x}^{-1} \left\{g_m \right\} (x)  = i \cdot \sqrt{ \dfrac{2}{\pi} } \cdot p.v. \ \dfrac
{1} { (-1)^m x} = i \cdot (-1)^m \cdot \sqrt{ \dfrac{2}{\pi} } \cdot p.v. \ \dfrac {1} { x} \in \mathcal{H}^{-1/2-\epsilon}\left(\R\right).
\end{equation*}
Using the notations $\delta_x$ in case the delta distribution is only applied in $x$-direction and $\delta_y$ accordingly, $v_x:=(p.v. \ \tfrac{1}{x})$ and $v_y:=(p.v. \ \tfrac{1}{y})$ as well as $v_{xy}:=(p.v. \ \tfrac{1}{xy})$, the $2$d $PSF_{pyr}^{mn} \in \mathcal{H}^{-1-\epsilon}\left(\R^2\right)$ is given by
\begin{eqnarray*}
PSF_{pyr}^{mn}
&=& \dfrac{\pi}{2} \cdot \delta_x \delta_y 
+   \dfrac{1}{2} \cdot i \cdot (-1)^m \cdot v_x\delta_y  
+  \dfrac{1}{2} \cdot i \cdot (-1)^n \cdot \delta_x  v_y
+  \dfrac{1}{2\pi} \cdot (-1)^{m+n+1} \cdot v_{xy}. \nonumber
\end{eqnarray*}
According to the standard description of optical systems, the complex amplitude $\psi_{det}$ coming to the detector plane is the inverse 2d Fourier transform of the complex amplitude after the pyramid $\left(\psi_{aper}^\epsilon\cdot OTF_{pyr}\right)$ which results (by application of the convolution theorem) in a convolution of the incoming complex amplitude $\psi^\epsilon_{aper}$ with the point spread function of the glass pyramid as described, e.g., in \cite{KoVe07}. This step can now mathematically be written in the sense of distributions as a shifted PSF distribution applied to the complex amplitude of the incoming phase
\begin{align*}
\psi_{det} (x,y) &=  \dfrac{ 1 } { 2\pi }  \langle PSF_{pyr}\left(\left(x,y\right)-\left(\cdot,\cdot\right)\right),\psi_{aper}^\epsilon \rangle.
\end{align*}    
Then, by linearity, the four independent beams $\psi_{det}^{mn} , m,n \in \left\{ 0,1 \right\} $, falling onto the detector are given by
{\footnotesize\begin{align*} \label{eq:psi_mn}
\psi_{det}^{mn} (x,y) &=  \dfrac{ 1 } { 2\pi }  \langle PSF^{mn}_{pyr}\left(\left(x,y\right)-\left(\cdot,\cdot\right)\right),\psi_{aper}^\epsilon \rangle \\
&= \dfrac{ 1 } { 2\pi }  \langle \left(       \dfrac{\pi}{2} \cdot  \delta_x  \delta_y + \dfrac{1}{2} \cdot i \cdot (-1)^m \cdot v_x \delta_y  
+  \dfrac{1}{2} \cdot i \cdot (-1)^n \cdot \delta_x  v_y +  \dfrac{1}{2\pi} \cdot (-1)^{m+n+1} \cdot v_{xy}\right)        \left(\left(x,y\right)-\left(\cdot,\cdot\right)\right),\psi_{aper}^\epsilon \rangle \\
&=\dfrac{1}{4} \langle \left(\delta_x\delta_y\right)\left(\left(x,y\right)-\left(\cdot,\cdot\right)\right),\psi_{aper}^\epsilon\rangle + 
\dfrac{\left(-1\right)^m i}{4\pi} \langle \left(v_x\delta_y\right)\left(\left(x,y\right)-\left(\cdot,\cdot\right)\right),\psi_{aper}^\epsilon\rangle \\
&+\dfrac{\left(-1\right)^n i}{4\pi} \langle \left(\delta_x v_y\right)\left(\left(x,y\right)-\left(\cdot,\cdot\right)\right),\psi_{aper}^\epsilon\rangle + 
\dfrac{\left(-1\right)^{m+n+1}}{4\pi^2} \langle v_{xy}\left(\left(x,y\right)-\left(\cdot,\cdot\right)\right),\psi_{aper}^\epsilon\rangle\\
&=\dfrac{1}{4} \psi_{aper}^\epsilon\left(x,y\right) + 
\dfrac{\left(-1\right)^m i}{4\pi} \langle v_x\left(x-\cdot\right),\psi_{aper}^\epsilon\left(\cdot,y\right)\rangle \\
&+\dfrac{\left(-1\right)^n i}{4\pi} \langle v_y\left(y-\cdot\right),\psi_{aper}^\epsilon\left(x,\cdot\right)\rangle + 
\dfrac{\left(-1\right)^{m+n+1}}{4\pi^2} \langle v_{xy}\left(\left(x,y\right)-\left(\cdot,\cdot\right)\right),\psi_{aper}^\epsilon\rangle.
\end{align*}  
}The four complex amplitudes are explicitly formulated as
\begin{align*}
\psi_{det}^{00} (x,y) &=
\dfrac{1}{4}\psi_{aper}^\epsilon\left(x,y\right) + 
\dfrac{ i}{4\pi} \langle v_x\left(x-\cdot\right),\psi_{aper}^\epsilon\left(\cdot,y\right)\rangle \\
&+\dfrac{ i}{4\pi} \langle v_y\left(y-\cdot\right),\psi_{aper}^\epsilon\left(x,\cdot\right)\rangle -
\dfrac{1}{4\pi^2} \langle v_{xy}\left(\left(x,y\right)-\left(\cdot,\cdot\right)\right),\psi_{aper}^\epsilon\rangle, 
\end{align*} 
\begin{align*}
\psi_{det}^{01} (x,y) &=
\dfrac{1}{4}\psi_{aper}^\epsilon\left(x,y\right) + 
\dfrac{ i}{4\pi} \langle v_x\left(x-\cdot\right),\psi_{aper}^\epsilon\left(\cdot,y\right)\rangle \\
&-\dfrac{ i}{4\pi} \langle v_y\left(y-\cdot\right),\psi_{aper}^\epsilon\left(x,\cdot\right)\rangle +
\dfrac{1}{4\pi^2} \langle v_{xy}\left(\left(x,y\right)-\left(\cdot,\cdot\right)\right),\psi_{aper}^\epsilon\rangle, 
\end{align*} 
\begin{align*}
\psi_{det}^{10} (x,y) &=
\dfrac{1}{4}\psi_{aper}^\epsilon\left(x,y\right) -
\dfrac{ i}{4\pi} \langle v_x\left(x-\cdot\right),\psi_{aper}^\epsilon\left(\cdot,y\right)\rangle \\
&+\dfrac{ i}{4\pi} \langle v_y\left(y-\cdot\right),\psi_{aper}^\epsilon\left(x,\cdot\right)\rangle +
\dfrac{1}{4\pi^2} \langle v_{xy}\left(\left(x,y\right)-\left(\cdot,\cdot\right)\right),\psi_{aper}^\epsilon\rangle, 
\end{align*} 
\begin{align*}
\psi_{det}^{11} (x,y) &=
\dfrac{1}{4}\psi_{aper}^\epsilon\left(x,y\right) -
\dfrac{ i}{4\pi} \langle v_x\left(x-\cdot\right),\psi_{aper}^\epsilon\left(\cdot,y\right)\rangle \\
&-\dfrac{ i}{4\pi} \langle v_y\left(y-\cdot\right),\psi_{aper}^\epsilon\left(x,\cdot\right)\rangle -
\dfrac{1}{4\pi^2} \langle v_{xy}\left(\left(x,y\right)-\left(\cdot,\cdot\right)\right),\psi_{aper}^\epsilon\rangle. 
\end{align*} 
Now, the intensities on the detector are computed as 
\be{intensity} \notag
I^{mn}\left(x,y\right) = \psi_{det}^{mn}\left(x,y\right)\overline{\psi_{det}^{mn}\left(x,y\right)}, \qquad \qquad m,n\in \{0,1\}.
\ee
If we abbreviate $\psi_{aper}^\epsilon$ by $\psi$ and omit the arguments for simplicity of notation the four intensities are evaluated as
\begin{eqnarray*}
I^{00}\left(x,y\right) &=& \psi_{det}^{00}\left(x,y\right)\overline{\psi_{det}^{00}\left(x,y\right)} \\
&=& \left[\dfrac{1}{4}\psi + 
\dfrac{ i}{4\pi} \langle v_x,\psi\rangle
+\dfrac{ i}{4\pi} \langle v_y,\psi\rangle -
\dfrac{1}{4\pi^2} \langle v_{xy},\psi\rangle\right] \\
& \cdot &
 \left[\dfrac{1}{4}\overline{\psi} -
\dfrac{ i}{4\pi} \langle v_x,\overline{\psi}\rangle
-\dfrac{ i}{4\pi} \langle v_y,\overline{\psi}\rangle -
\dfrac{1}{4\pi^2} \langle v_{xy},\overline{\psi}\rangle\right]  \\
&=&
\dfrac{1}{16}\psi\overline{\psi} - \dfrac{i}{16\pi}\psi\langle v_x,\overline{\psi}\rangle- \dfrac{i}{16\pi}\psi\langle v_y,\overline{\psi}\rangle - \dfrac{1}{16\pi^2}\psi\langle v_{xy},\overline{\psi}\rangle \\
&+&
\dfrac{i}{16\pi}\langle v_x,\psi\rangle \overline{\psi} + \dfrac{1}{16\pi^2}\langle v_x,\psi\rangle \langle v_x,\overline{\psi}\rangle+\dfrac{1}{16\pi^2}\langle v_x,\psi\rangle \langle v_y,\overline{\psi}\rangle - \dfrac{i}{16\pi^3}\langle v_x,\psi\rangle \langle v_{xy},\overline{\psi}\rangle \\
&+&
\dfrac{i}{16\pi}\langle v_y,\psi\rangle \overline{\psi} + \dfrac{1}{16\pi^2}\langle v_y,\psi\rangle \langle v_x,\overline{\psi}\rangle+\dfrac{1}{16\pi^2}\langle v_y,\psi\rangle \langle v_y,\overline{\psi}\rangle - \dfrac{i}{16\pi^3}\langle v_y,\psi\rangle \langle v_{xy},\overline{\psi}\rangle \\
&-&
\dfrac{1}{16\pi^2}\langle v_{xy},\psi\rangle \overline{\psi} + \dfrac{i}{16\pi^3}\langle v_{xy},\psi\rangle \langle v_x,\overline{\psi}\rangle+\dfrac{i}{16\pi^3}\langle v_{xy},\psi\rangle \langle v_y,\overline{\psi}\rangle + \dfrac{1}{16\pi^4}\langle v_{xy},\psi\rangle \langle v_{xy},\overline{\psi}\rangle,
\end{eqnarray*} 
\begin{eqnarray*}
I^{01}\left(x,y\right) &=& \psi_{det}^{01}\left(x,y\right)\overline{\psi_{det}^{01}\left(x,y\right)} \\
&=& \left[\dfrac{1}{4}\psi + 
\dfrac{ i}{4\pi} \langle v_x,\psi\rangle
-\dfrac{ i}{4\pi} \langle v_y,\psi\rangle +
\dfrac{1}{4\pi^2} \langle v_{xy},\psi\rangle\right]  \\
&\cdot&
 \left[\dfrac{1}{4}\overline{\psi} -
\dfrac{ i}{4\pi} \langle v_x,\overline{\psi}\rangle
+\dfrac{ i}{4\pi} \langle v_y,\overline{\psi}\rangle +
\dfrac{1}{4\pi^2} \langle v_{xy},\overline{\psi}\rangle\right]  \\
&=&
\dfrac{1}{16}\psi\overline{\psi} - \dfrac{i}{16\pi}\psi\langle v_x,\overline{\psi}\rangle+ \dfrac{i}{16\pi}\psi\langle v_y,\overline{\psi}\rangle + \dfrac{1}{16\pi^2}\psi\langle v_{xy},\overline{\psi}\rangle \\
&+&
\dfrac{i}{16\pi}\langle v_x,\psi\rangle \overline{\psi} + \dfrac{1}{16\pi^2}\langle v_x,\psi\rangle \langle v_x,\overline{\psi}\rangle-\dfrac{1}{16\pi^2}\langle v_x,\psi\rangle \langle v_y,\overline{\psi}\rangle + \dfrac{i}{16\pi^3}\langle v_x,\psi\rangle \langle v_{xy},\overline{\psi}\rangle \\
&-&
\dfrac{i}{16\pi}\langle v_y,\psi\rangle \overline{\psi} - \dfrac{1}{16\pi^2}\langle v_y,\psi\rangle \langle v_x,\overline{\psi}\rangle+\dfrac{1}{16\pi^2}\langle v_y,\psi\rangle \langle v_y,\overline{\psi}\rangle - \dfrac{i}{16\pi^3}\langle v_y,\psi\rangle \langle v_{xy},\overline{\psi}\rangle \\
&+&
\dfrac{1}{16\pi^2}\langle v_{xy},\psi\rangle \overline{\psi} - \dfrac{i}{16\pi^3}\langle v_{xy},\psi\rangle \langle v_x,\overline{\psi}\rangle+\dfrac{i}{16\pi^3}\langle v_{xy},\psi\rangle \langle v_y,\overline{\psi}\rangle + \dfrac{1}{16\pi^4}\langle v_{xy},\psi\rangle \langle v_{xy},\overline{\psi}\rangle,
\end{eqnarray*}
\begin{eqnarray*}
I^{10}\left(x,y\right) &=& \psi_{det}^{10}\left(x,y\right)\overline{\psi_{det}^{10}\left(x,y\right)} \\
&=& \left[\dfrac{1}{4}\psi - 
\dfrac{ i}{4\pi} \langle v_x,\psi\rangle
+\dfrac{ i}{4\pi} \langle v_y,\psi\rangle +
\dfrac{1}{4\pi^2} \langle v_{xy},\psi\rangle\right] \\
&\cdot&
 \left[\dfrac{1}{4}\overline{\psi} +
\dfrac{ i}{4\pi} \langle v_x,\overline{\psi}\rangle
-\dfrac{ i}{4\pi} \langle v_y,\overline{\psi}\rangle +
\dfrac{1}{4\pi^2} \langle v_{xy},\overline{\psi}\rangle\right]  \\
&=&
\dfrac{1}{16}\psi\overline{\psi} + \dfrac{i}{16\pi}\psi\langle v_x,\overline{\psi}\rangle- \dfrac{i}{16\pi}\psi\langle v_y,\overline{\psi}\rangle + \dfrac{1}{16\pi^2}\psi\langle v_{xy},\overline{\psi}\rangle \\
&-&
\dfrac{i}{16\pi}\langle v_x,\psi\rangle \overline{\psi} + \dfrac{1}{16\pi^2}\langle v_x,\psi\rangle \langle v_x,\overline{\psi}\rangle-\dfrac{1}{16\pi^2}\langle v_x,\psi\rangle \langle v_y,\overline{\psi}\rangle - \dfrac{i}{16\pi^3}\langle v_x,\psi\rangle \langle v_{xy},\overline{\psi}\rangle \\
&+&
\dfrac{i}{16\pi}\langle v_y,\psi\rangle \overline{\psi} - \dfrac{1}{16\pi^2}\langle v_y,\psi\rangle \langle v_x,\overline{\psi}\rangle+\dfrac{1}{16\pi^2}\langle v_y,\psi\rangle \langle v_y,\overline{\psi}\rangle + \dfrac{i}{16\pi^3}\langle v_y,\psi\rangle \langle v_{xy},\overline{\psi}\rangle \\
&+&
\dfrac{1}{16\pi^2}\langle v_{xy},\psi\rangle \overline{\psi} + \dfrac{i}{16\pi^3}\langle v_{xy},\psi\rangle \langle v_x,\overline{\psi}\rangle-\dfrac{i}{16\pi^3}\langle v_{xy},\psi\rangle \langle v_y,\overline{\psi}\rangle + \dfrac{1}{16\pi^4}\langle v_{xy},\psi\rangle \langle v_{xy},\overline{\psi}\rangle,
\end{eqnarray*}
\begin{eqnarray*}
I^{11}\left(x,y\right) &=& \psi_{det}^{11}\left(x,y\right)\overline{\psi_{det}^{11}\left(x,y\right)} \\
&=& \left[\dfrac{1}{4}\psi - 
\dfrac{ i}{4\pi} \langle v_x,\psi\rangle
-\dfrac{ i}{4\pi} \langle v_y,\psi\rangle -
\dfrac{1}{4\pi^2} \langle v_{xy},\psi\rangle\right] \\
&\cdot&
 \left[\dfrac{1}{4}\overline{\psi} +
\dfrac{ i}{4\pi} \langle v_x,\overline{\psi}\rangle
+\dfrac{ i}{4\pi} \langle v_y,\overline{\psi}\rangle -
\dfrac{1}{4\pi^2} \langle v_{xy},\overline{\psi}\rangle\right]  \\
&=&
\dfrac{1}{16}\psi\overline{\psi} + \dfrac{i}{16\pi}\psi\langle v_x,\overline{\psi}\rangle+ \dfrac{i}{16\pi}\psi\langle v_y,\overline{\psi}\rangle - \dfrac{1}{16\pi^2}\psi\langle v_{xy},\overline{\psi}\rangle \\
&-&
\dfrac{i}{16\pi}\langle v_x,\psi\rangle \overline{\psi} + \dfrac{1}{16\pi^2}\langle v_x,\psi\rangle \langle v_x,\overline{\psi}\rangle+\dfrac{1}{16\pi^2}\langle v_x,\psi\rangle \langle v_y,\overline{\psi}\rangle + \dfrac{i}{16\pi^3}\langle v_x,\psi\rangle \langle v_{xy},\overline{\psi}\rangle \\
&-&
\dfrac{i}{16\pi}\langle v_y,\psi\rangle \overline{\psi} + \dfrac{1}{16\pi^2}\langle v_y,\psi\rangle \langle v_x,\overline{\psi}\rangle+\dfrac{1}{16\pi^2}\langle v_y,\psi\rangle \langle v_y,\overline{\psi}\rangle + \dfrac{i}{16\pi^3}\langle v_y,\psi\rangle \langle v_{xy},\overline{\psi}\rangle \\
&-&
\dfrac{1}{16\pi^2}\langle v_{xy},\psi\rangle \overline{\psi} - \dfrac{i}{16\pi^3}\langle v_{xy},\psi\rangle \langle v_x,\overline{\psi}\rangle-\dfrac{i}{16\pi^3}\langle v_{xy},\psi\rangle \langle v_y,\overline{\psi}\rangle + \dfrac{1}{16\pi^4}\langle v_{xy},\psi\rangle \langle v_{xy},\overline{\psi}\rangle.
\end{eqnarray*}

Taking the sums according to~\eqref{meas_equ}, we obtain the non-modulated (indicated by the superscript $n$) pyramid sensor data $s_x^n$ in $x$-direction as
\begin{align*}
I_0\cdot s_x^n (x,y) &= [I^{01}(x,y) + I^{00}(x,y)] - [I^{11}(x,y) + I^{10}(x,y)]  \\
                 &=
\dfrac{1}{16}\psi\overline{\psi} - \dfrac{i}{16\pi}\psi\langle v_x,\overline{\psi}\rangle+ \dfrac{i}{16\pi}\psi\langle v_y,\overline{\psi}\rangle + \dfrac{1}{16\pi^2}\psi\langle v_{xy},\overline{\psi}\rangle \\
&+
\dfrac{i}{16\pi}\langle v_x,\psi\rangle \overline{\psi} + \dfrac{1}{16\pi^2}\langle v_x,\psi\rangle \langle v_x,\overline{\psi}\rangle-\dfrac{1}{16\pi^2}\langle v_x,\psi\rangle \langle v_y,\overline{\psi}\rangle + \dfrac{i}{16\pi^3}\langle v_x,\psi\rangle \langle v_{xy},\overline{\psi}\rangle \\
&-
\dfrac{i}{16\pi}\langle v_y,\psi\rangle \overline{\psi} - \dfrac{1}{16\pi^2}\langle v_y,\psi\rangle \langle v_x,\overline{\psi}\rangle+\dfrac{1}{16\pi^2}\langle v_y,\psi\rangle \langle v_y,\overline{\psi}\rangle - \dfrac{i}{16\pi^3}\langle v_y,\psi\rangle \langle v_{xy},\overline{\psi}\rangle \\
&+
\dfrac{1}{16\pi^2}\langle v_{xy},\psi\rangle \overline{\psi} - \dfrac{i}{16\pi^3}\langle v_{xy},\psi\rangle \langle v_x,\overline{\psi}\rangle+\dfrac{i}{16\pi^3}\langle v_{xy},\psi\rangle \langle v_y,\overline{\psi}\rangle + \dfrac{1}{16\pi^4}\langle v_{xy},\psi\rangle \langle v_{xy},\overline{\psi}\rangle \\
                 &+
\dfrac{1}{16}\psi\overline{\psi} - \dfrac{i}{16\pi}\psi\langle v_x,\overline{\psi}\rangle- \dfrac{i}{16\pi}\psi\langle v_y,\overline{\psi}\rangle - \dfrac{1}{16\pi^2}\psi\langle v_{xy},\overline{\psi}\rangle \\
&+
\dfrac{i}{16\pi}\langle v_x,\psi\rangle \overline{\psi} + \dfrac{1}{16\pi^2}\langle v_x,\psi\rangle \langle v_x,\overline{\psi}\rangle+\dfrac{1}{16\pi^2}\langle v_x,\psi\rangle \langle v_y,\overline{\psi}\rangle - \dfrac{i}{16\pi^3}\langle v_x,\psi\rangle \langle v_{xy},\overline{\psi}\rangle \\
&+
\dfrac{i}{16\pi}\langle v_y,\psi\rangle \overline{\psi} + \dfrac{1}{16\pi^2}\langle v_y,\psi\rangle \langle v_x,\overline{\psi}\rangle+\dfrac{1}{16\pi^2}\langle v_y,\psi\rangle \langle v_y,\overline{\psi}\rangle - \dfrac{i}{16\pi^3}\langle v_y,\psi\rangle \langle v_{xy},\overline{\psi}\rangle \\
&-
\dfrac{1}{16\pi^2}\langle v_{xy},\psi\rangle \overline{\psi} + \dfrac{i}{16\pi^3}\langle v_{xy},\psi\rangle \langle v_x,\overline{\psi}\rangle+\dfrac{i}{16\pi^3}\langle v_{xy},\psi\rangle \langle v_y,\overline{\psi}\rangle + \dfrac{1}{16\pi^4}\langle v_{xy},\psi\rangle \langle v_{xy},\overline{\psi}\rangle \\
                  &-   
\dfrac{1}{16}\psi\overline{\psi} - \dfrac{i}{16\pi}\psi\langle v_x,\overline{\psi}\rangle- \dfrac{i}{16\pi}\psi\langle v_y,\overline{\psi}\rangle + \dfrac{1}{16\pi^2}\psi\langle v_{xy},\overline{\psi}\rangle \\
&+
\dfrac{i}{16\pi}\langle v_x,\psi\rangle \overline{\psi} - \dfrac{1}{16\pi^2}\langle v_x,\psi\rangle \langle v_x,\overline{\psi}\rangle-\dfrac{1}{16\pi^2}\langle v_x,\psi\rangle \langle v_y,\overline{\psi}\rangle - \dfrac{i}{16\pi^3}\langle v_x,\psi\rangle \langle v_{xy},\overline{\psi}\rangle \\
&+
\dfrac{i}{16\pi}\langle v_y,\psi\rangle \overline{\psi} - \dfrac{1}{16\pi^2}\langle v_y,\psi\rangle \langle v_x,\overline{\psi}\rangle-\dfrac{1}{16\pi^2}\langle v_y,\psi\rangle \langle v_y,\overline{\psi}\rangle - \dfrac{i}{16\pi^3}\langle v_y,\psi\rangle \langle v_{xy},\overline{\psi}\rangle \\
&+
\dfrac{1}{16\pi^2}\langle v_{xy},\psi\rangle \overline{\psi} + \dfrac{i}{16\pi^3}\langle v_{xy},\psi\rangle \langle v_x,\overline{\psi}\rangle+\dfrac{i}{16\pi^3}\langle v_{xy},\psi\rangle \langle v_y,\overline{\psi}\rangle - \dfrac{1}{16\pi^4}\langle v_{xy},\psi\rangle \langle v_{xy},\overline{\psi}\rangle \\
                  &-     
\dfrac{1}{16}\psi\overline{\psi} - \dfrac{i}{16\pi}\psi\langle v_x,\overline{\psi}\rangle+ \dfrac{i}{16\pi}\psi\langle v_y,\overline{\psi}\rangle - \dfrac{1}{16\pi^2}\psi\langle v_{xy},\overline{\psi}\rangle \\
&+
\dfrac{i}{16\pi}\langle v_x,\psi\rangle \overline{\psi} - \dfrac{1}{16\pi^2}\langle v_x,\psi\rangle \langle v_x,\overline{\psi}\rangle+\dfrac{1}{16\pi^2}\langle v_x,\psi\rangle \langle v_y,\overline{\psi}\rangle + \dfrac{i}{16\pi^3}\langle v_x,\psi\rangle \langle v_{xy},\overline{\psi}\rangle \\
&-
\dfrac{i}{16\pi}\langle v_y,\psi\rangle \overline{\psi} + \dfrac{1}{16\pi^2}\langle v_y,\psi\rangle \langle v_x,\overline{\psi}\rangle-\dfrac{1}{16\pi^2}\langle v_y,\psi\rangle \langle v_y,\overline{\psi}\rangle - \dfrac{i}{16\pi^3}\langle v_y,\psi\rangle \langle v_{xy},\overline{\psi}\rangle \\
&-
\dfrac{1}{16\pi^2}\langle v_{xy},\psi\rangle \overline{\psi} - \dfrac{i}{16\pi^3}\langle v_{xy},\psi\rangle \langle v_x,\overline{\psi}\rangle+\dfrac{i}{16\pi^3}\langle v_{xy},\psi\rangle \langle v_y,\overline{\psi}\rangle - \dfrac{1}{16\pi^4}\langle v_{xy},\psi\rangle \langle v_{xy},\overline{\psi}\rangle   ,   
\end{align*}   
which simplifies to
\begin{align*}
I_0\cdot s_x^n (x,y) &=
-\dfrac{i}{4\pi}\psi\langle v_x,\overline{\psi}\rangle  +\dfrac{i}{4\pi}\langle v_x,\psi\rangle\overline{\psi} - \dfrac{i}{4\pi^3}\langle v_y,\psi\rangle\langle v_{xy},\overline{\psi}\rangle+ \dfrac{i}{4\pi^3}\langle v_{xy},\psi\rangle\langle v_y,\overline{\psi}\rangle    \\
                                       &=
-\dfrac{i}{4\pi}\left[\psi\langle v_x,\overline{\psi}\rangle -\langle v_x,\psi\rangle\overline{\psi}\right] - \dfrac{i}{4\pi^3}\left[\langle v_y,\psi\rangle\langle v_{xy},\overline{\psi}\rangle - \langle v_{xy},\psi\rangle\langle v_y,\overline{\psi}\rangle  \right]     \\                                                  
                                       &=
-\dfrac{i}{4\pi}\left[\psi\left(x,y\right)\langle v_x\left(x-\cdot\right),\overline{\psi\left(\cdot,y\right)}\rangle -\langle v_x\left(x-\cdot\right),\psi\left(\cdot,y\right)\rangle\overline{\psi\left(x,y\right)}\right] \\
&- \dfrac{i}{4\pi^3}\left[\langle v_y\left(y-\cdot\right),\psi\left(x,\cdot\right)\rangle\langle v_{xy}\left(x-\cdot,y-\cdot\right),\overline{\psi}\rangle - \langle v_{xy}\left(x-\cdot,y-\cdot\right),\psi\rangle\langle v_y\left(y-\cdot\right),\overline{\psi\left(x,\cdot\right)}\rangle  \right].
\end{align*}   
This can further be formulated as
\begin{align*}
I_0\cdot s_x^n (x,y) &= -\dfrac{i}{4\pi}\left[\mathcal{X}_{\Omega}\left(x,y\right) \cdot \exp \left(-i\Phi\left(x,y\right)\right) \ p.v. \int\limits_\R \mathcal{X}_{\Omega}\left(x',y\right)\cdot \exp\left(i\Phi\left(x',y\right)\right)\dfrac{1}{x-x'} \ dx' \right. \\
&-\left. \mathcal{X}_{\Omega}\left(x,y\right)\cdot\exp\left(i\Phi\left(x,y\right)\right) \ p.v. \int\limits_\R \mathcal{X}_{\Omega}\left(x',y\right)\cdot \exp\left(-i\Phi\left(x',y\right)\right) \dfrac{1}{x-x'} \ dx'  \right]\\
&-\dfrac{i}{4\pi^3}\left[ p.v. \int\limits_\R \mathcal{X}_{\Omega}\left(x,y''\right)\cdot \exp\left(-i\Phi\left(x,y''\right)\right)\dfrac{1}{y-y''} \ dy'' \right. \\
&\cdot \left. p.v. \int\limits_\R  p.v. \int\limits_\R \mathcal{X}_{\Omega}\left(x',y'\right)\cdot \exp\left(i\Phi\left(x',y'\right)\right)\dfrac{1}{\left(x-x'\right)\left(y-y'\right)} \ dy' \ dx'  \right. \\
&-\left. p.v. \int\limits_\R \mathcal{X}_{\Omega}\left(x,y''\right)\cdot \exp\left(i\Phi\left(x,y''\right)\right)\dfrac{1}{y-y''} \ dy'' \right. \\
&\cdot \left.  p.v. \int\limits_\R  p.v. \int\limits_\R \mathcal{X}_{\Omega}\left(x',y'\right)\cdot \exp\left(-i\Phi\left(x',y'\right)\right)\dfrac{1}{\left(x-x'\right)\left(y-y'\right)} \ dy' \ dx'  \right] \\
&=
-\dfrac{i}{4\pi}\left[\mathcal{X}_{\Omega}\left(x,y\right) \cdot \exp \left(-i\Phi\left(x,y\right)\right) \ p.v. \int\limits_{\Omega_y} \exp\left(i\Phi\left(x',y\right)\right)\dfrac{1}{x-x'} \ dx' \right. \\
&-\left. \mathcal{X}_{\Omega}\left(x,y\right)\cdot\exp\left(i\Phi\left(x,y\right)\right) \ p.v. \int\limits_{\Omega_y} \exp\left(-i\Phi\left(x',y\right)\right) \dfrac{1}{x-x'} \ dx'  \right]\\
&-\dfrac{i}{4\pi^3}\left[\mathcal{X}_{\Omega_y}\left(x\right) \ p.v. \int\limits_{\Omega_x} \exp\left(-i\Phi\left(x,y''\right)\right)\dfrac{1}{y-y''} \ dy'' \right. \\
&\cdot \left. p.v. \int\limits_{\Omega_y}  p.v. \int\limits_{\Omega_x} \exp\left(i\Phi\left(x',y'\right)\right)\dfrac{1}{\left(x-x'\right)\left(y-y'\right)} \ dy' \ dx'  \right. \\
&-\left. \mathcal{X}_{\Omega_y}\left(x\right)p.v. \int\limits_{\Omega_x} \exp\left(i\Phi\left(x,y''\right)\right)\dfrac{1}{y-y''} \ dy'' \right. \\
&\cdot \left.  p.v. \int\limits_{\Omega_y}  p.v. \int\limits_{\Omega_x}\exp\left(-i\Phi\left(x',y'\right)\right)\dfrac{1}{\left(x-x'\right)\left(y-y'\right)} \ dy' \ dx'  \right]. 
\end{align*}
With Euler's  and trigonometric formulas we obtain
\begin{align*}
I_0\cdot s_x^n (x,y)
={}&- \dfrac{i}{4\pi}\mathcal{X}_{\Omega}\left(x,y\right) \left[2i \ p.v. \int\limits_{\Omega_y}\dfrac{\sin\left[\Phi\left(x',y\right)-\Phi\left(x,y\right)\right]}{x-x'} \ dx'  \right] \\
&-
\dfrac{i}{4\pi^3}\mathcal{X}_{\Omega_y}\left(x\right)\left[2i \ p.v. \int\limits_{\Omega_x}  \ p.v. \int\limits_{\Omega_y}  \ p.v. \int\limits_{\Omega_x} \dfrac{\sin\left[\Phi\left(x',y'\right)-\Phi\left(x,y''\right)\right]}{\left(x-x'\right)\left(y-y'\right)\left(y-y''\right)} \ dy'  \ dx' \ dy''
\right] \\
={}&
\mathcal{X}_{\Omega}\left(x,y\right)\dfrac{1}{2\pi} \ p.v. \int\limits_{\Omega_y}\dfrac{\sin\left[\Phi\left(x',y\right)-\Phi\left(x,y\right)\right]}{x-x'} \ dx' \\
&+
\mathcal{X}_{\Omega_y}\left(x\right)\dfrac{1}{2\pi^3} \ p.v. \int\limits_{\Omega_y}  \ p.v. \int\limits_{\Omega_x}  \ p.v. \int\limits_{\Omega_x} \dfrac{\sin\left[\Phi\left(x',y'\right)-\Phi\left(x,y''\right)\right]}{\left(x-x'\right)\left(y-y'\right)\left(y-y''\right)} \ dy''  \ dy' \ dx'.
\end{align*}

Taking the sums according to~\eqref{meas_equ}, the non-modulated pyramid sensor data $s_y^n$ in $y$-direction are written as
\begin{align*}
I_0\cdot s_y^n (x,y) &= [I^{01}(x,y) + I^{11}(x,y)] - [I^{00}(x,y) + I^{10}(x,y)]  \\
                 &=
\dfrac{1}{16}\psi\overline{\psi} - \dfrac{i}{16\pi}\psi\langle v_x,\overline{\psi}\rangle+ \dfrac{i}{16\pi}\psi\langle v_y,\overline{\psi}\rangle + \dfrac{1}{16\pi^2}\psi\langle v_{xy},\overline{\psi}\rangle \\
&+
\dfrac{i}{16\pi}\langle v_x,\psi\rangle \overline{\psi} + \dfrac{1}{16\pi^2}\langle v_x,\psi\rangle \langle v_x,\overline{\psi}\rangle-\dfrac{1}{16\pi^2}\langle v_x,\psi\rangle \langle v_y,\overline{\psi}\rangle + \dfrac{i}{16\pi^3}\langle v_x,\psi\rangle \langle v_{xy},\overline{\psi}\rangle \\
&-
\dfrac{i}{16\pi}\langle v_y,\psi\rangle \overline{\psi} - \dfrac{1}{16\pi^2}\langle v_y,\psi\rangle \langle v_x,\overline{\psi}\rangle+\dfrac{1}{16\pi^2}\langle v_y,\psi\rangle \langle v_y,\overline{\psi}\rangle - \dfrac{i}{16\pi^3}\langle v_y,\psi\rangle \langle v_{xy},\overline{\psi}\rangle \\
&+
\dfrac{1}{16\pi^2}\langle v_{xy},\psi\rangle \overline{\psi} - \dfrac{i}{16\pi^3}\langle v_{xy},\psi\rangle \langle v_x,\overline{\psi}\rangle+\dfrac{i}{16\pi^3}\langle v_{xy},\psi\rangle \langle v_y,\overline{\psi}\rangle + \dfrac{1}{16\pi^4}\langle v_{xy},\psi\rangle \langle v_{xy},\overline{\psi}\rangle \\
                 &+ 
\dfrac{1}{16}\psi\overline{\psi} + \dfrac{i}{16\pi}\psi\langle v_x,\overline{\psi}\rangle+ \dfrac{i}{16\pi}\psi\langle v_y,\overline{\psi}\rangle - \dfrac{1}{16\pi^2}\psi\langle v_{xy},\overline{\psi}\rangle \\
&-
\dfrac{i}{16\pi}\langle v_x,\psi\rangle \overline{\psi} + \dfrac{1}{16\pi^2}\langle v_x,\psi\rangle \langle v_x,\overline{\psi}\rangle+\dfrac{1}{16\pi^2}\langle v_x,\psi\rangle \langle v_y,\overline{\psi}\rangle + \dfrac{i}{16\pi^3}\langle v_x,\psi\rangle \langle v_{xy},\overline{\psi}\rangle \\
&-
\dfrac{i}{16\pi}\langle v_y,\psi\rangle \overline{\psi} + \dfrac{1}{16\pi^2}\langle v_y,\psi\rangle \langle v_x,\overline{\psi}\rangle+\dfrac{1}{16\pi^2}\langle v_y,\psi\rangle \langle v_y,\overline{\psi}\rangle + \dfrac{i}{16\pi^3}\langle v_y,\psi\rangle \langle v_{xy},\overline{\psi}\rangle \\
&-
\dfrac{1}{16\pi^2}\langle v_{xy},\psi\rangle \overline{\psi} - \dfrac{i}{16\pi^3}\langle v_{xy},\psi\rangle \langle v_x,\overline{\psi}\rangle-\dfrac{i}{16\pi^3}\langle v_{xy},\psi\rangle \langle v_y,\overline{\psi}\rangle + \dfrac{1}{16\pi^4}\langle v_{xy},\psi\rangle \langle v_{xy},\overline{\psi}\rangle \\
                 &-
\dfrac{1}{16}\psi\overline{\psi} + \dfrac{i}{16\pi}\psi\langle v_x,\overline{\psi}\rangle+ \dfrac{i}{16\pi}\psi\langle v_y,\overline{\psi}\rangle + \dfrac{1}{16\pi^2}\psi\langle v_{xy},\overline{\psi}\rangle \\
&-
\dfrac{i}{16\pi}\langle v_x,\psi\rangle \overline{\psi} - \dfrac{1}{16\pi^2}\langle v_x,\psi\rangle \langle v_x,\overline{\psi}\rangle-\dfrac{1}{16\pi^2}\langle v_x,\psi\rangle \langle v_y,\overline{\psi}\rangle + \dfrac{i}{16\pi^3}\langle v_x,\psi\rangle \langle v_{xy},\overline{\psi}\rangle \\
&-
\dfrac{i}{16\pi}\langle v_y,\psi\rangle \overline{\psi} - \dfrac{1}{16\pi^2}\langle v_y,\psi\rangle \langle v_x,\overline{\psi}\rangle-\dfrac{1}{16\pi^2}\langle v_y,\psi\rangle \langle v_y,\overline{\psi}\rangle + \dfrac{i}{16\pi^3}\langle v_y,\psi\rangle \langle v_{xy},\overline{\psi}\rangle \\
&+
\dfrac{1}{16\pi^2}\langle v_{xy},\psi\rangle \overline{\psi} - \dfrac{i}{16\pi^3}\langle v_{xy},\psi\rangle \langle v_x,\overline{\psi}\rangle-\dfrac{i}{16\pi^3}\langle v_{xy},\psi\rangle \langle v_y,\overline{\psi}\rangle - \dfrac{1}{16\pi^4}\langle v_{xy},\psi\rangle \langle v_{xy},\overline{\psi}\rangle \\
                 &-
\dfrac{1}{16}\psi\overline{\psi} - \dfrac{i}{16\pi}\psi\langle v_x,\overline{\psi}\rangle+ \dfrac{i}{16\pi}\psi\langle v_y,\overline{\psi}\rangle - \dfrac{1}{16\pi^2}\psi\langle v_{xy},\overline{\psi}\rangle \\
&+
\dfrac{i}{16\pi}\langle v_x,\psi\rangle \overline{\psi} - \dfrac{1}{16\pi^2}\langle v_x,\psi\rangle \langle v_x,\overline{\psi}\rangle+\dfrac{1}{16\pi^2}\langle v_x,\psi\rangle \langle v_y,\overline{\psi}\rangle + \dfrac{i}{16\pi^3}\langle v_x,\psi\rangle \langle v_{xy},\overline{\psi}\rangle \\
&-
\dfrac{i}{16\pi}\langle v_y,\psi\rangle \overline{\psi} + \dfrac{1}{16\pi^2}\langle v_y,\psi\rangle \langle v_x,\overline{\psi}\rangle-\dfrac{1}{16\pi^2}\langle v_y,\psi\rangle \langle v_y,\overline{\psi}\rangle - \dfrac{i}{16\pi^3}\langle v_y,\psi\rangle \langle v_{xy},\overline{\psi}\rangle \\
&-
\dfrac{1}{16\pi^2}\langle v_{xy},\psi\rangle \overline{\psi} - \dfrac{i}{16\pi^3}\langle v_{xy},\psi\rangle \langle v_x,\overline{\psi}\rangle+\dfrac{i}{16\pi^3}\langle v_{xy},\psi\rangle \langle v_y,\overline{\psi}\rangle - \dfrac{1}{16\pi^4}\langle v_{xy},\psi\rangle \langle v_{xy},\overline{\psi}\rangle,   
\end{align*} 
which is equivalent to
\begin{align*}
I_0\cdot s_y^n (x,y) &=
\dfrac{i}{4\pi}\psi\langle v_y,\overline{\psi}\rangle + \dfrac{i}{4\pi^3}\langle v_x,\psi\rangle\langle v_{xy},\overline{\psi}\rangle -\dfrac{i}{4\pi}\langle v_y,\psi\rangle\overline{\psi} - \dfrac{i}{4\pi^3}\langle v_{xy},\psi\rangle\langle v_x,\overline{\psi}\rangle    \\
                                       &=
\dfrac{i}{4\pi}\left[\psi\langle v_y,\overline{\psi}\rangle -\langle v_y,\psi\rangle\overline{\psi}\right] + \dfrac{i}{4\pi^3}\left[\langle v_x,\psi\rangle\langle v_{xy},\overline{\psi}\rangle - \langle v_{xy},\psi\rangle\langle v_x,\overline{\psi}\rangle  \right]     \\                                                  
                                       &=
\dfrac{i}{4\pi}\left[\psi\left(x,y\right)\langle v_y\left(y-\cdot\right),\overline{\psi\left(x,\cdot\right)}\rangle -\langle v_y\left(y-\cdot\right),\psi\left(x,\cdot\right)\rangle\overline{\psi\left(x,y\right)}\right] \\
&+ \dfrac{i}{4\pi^3}\left[\langle v_x\left(x-\cdot\right),\psi\left(\cdot,y\right)\rangle\langle v_{xy}\left(x-\cdot,y-\cdot\right),\overline{\psi}\rangle - \langle v_{xy}\left(x-\cdot,y-\cdot\right),\psi\rangle\langle v_x\left(x-\cdot\right),\overline{\psi\left(\cdot,y\right)}\rangle  \right]                                                       
\end{align*} 
and results in
\begin{align*}
I_0\cdot s_y^n (x,y) &= \dfrac{i}{4\pi}\left[\mathcal{X}_{\Omega}\left(x,y\right) \cdot \exp \left(-i\Phi\left(x,y\right)\right) \ p.v. \int\limits_\R \mathcal{X}_{\Omega}\left(x,y'\right)\cdot \exp\left(i\Phi\left(x,y'\right)\right)\dfrac{1}{y-y'} \ dy' \right. \\
&-\left. \mathcal{X}_{\Omega}\left(x,y\right)\cdot\exp\left(i\Phi\left(x,y\right)\right) \ p.v. \int\limits_\R \mathcal{X}_{\Omega}\left(x,y'\right)\cdot \exp\left(-i\Phi\left(x,y'\right)\right) \dfrac{1}{y-y'} \ dy'  \right]\\
&+\dfrac{i}{4\pi^3}\left[ p.v. \int\limits_\R \mathcal{X}_{\Omega}\left(x'',y\right)\cdot \exp\left(-i\Phi\left(x'',y\right)\right)\dfrac{1}{x-x''} \ dx'' \right. \\
&\cdot \left. p.v. \int\limits_\R  p.v. \int\limits_\R \mathcal{X}_{\Omega}\left(x',y'\right)\cdot \exp\left(i\Phi\left(x',y'\right)\right)\dfrac{1}{\left(x-x'\right)\left(y-y'\right)} \ dy' \ dx'  \right. \\
&-\left. p.v. \int\limits_\R \mathcal{X}_{\Omega}\left(x'',y\right)\cdot \exp\left(i\Phi\left(x'',y\right)\right)\dfrac{1}{x-x''} \ dx'' \right. \\
&\cdot \left.  p.v. \int\limits_\R  p.v. \int\limits_\R \mathcal{X}_{\Omega}\left(x',y'\right)\cdot \exp\left(-i\Phi\left(x',y'\right)\right)\dfrac{1}{\left(x-x'\right)\left(y-y'\right)} \ dy' \ dx'  \right] \\
&=
\dfrac{i}{4\pi}\left[\mathcal{X}_{\Omega}\left(x,y\right) \cdot \exp \left(-i\Phi\left(x,y\right)\right) \ p.v. \int\limits_{\Omega_x} \exp\left(i\Phi\left(x,y'\right)\right)\dfrac{1}{y-y'} \ dy' \right. \\
&-\left. \mathcal{X}_{\Omega}\left(x,y\right)\cdot\exp\left(i\Phi\left(x,y\right)\right) \ p.v. \int\limits_{\Omega_x} \exp\left(-i\Phi\left(x,y'\right)\right) \dfrac{1}{y-y'} \ dy'  \right]\\
&+\dfrac{i}{4\pi^3}\left[\mathcal{X}_{\Omega_x}\left(y\right) \ p.v. \int\limits_{\Omega_y} \exp\left(-i\Phi\left(x'',y\right)\right)\dfrac{1}{x-x''} \ dx'' \right. \\
&\cdot \left. p.v. \int\limits_{\Omega_y}  p.v. \int\limits_{\Omega_x} \exp\left(i\Phi\left(x',y'\right)\right)\dfrac{1}{\left(x-x'\right)\left(y-y'\right)} \ dy' \ dx'  \right. \\
&-\left. \mathcal{X}_{\Omega_x}\left(y\right)p.v. \int\limits_{\Omega_y} \exp\left(i\Phi\left(x'',y\right)\right)\dfrac{1}{x-x''} \ dx'' \right. \\
&\cdot \left.  p.v. \int\limits_{\Omega_y}  p.v. \int\limits_{\Omega_x}\exp\left(-i\Phi\left(x',y'\right)\right)\dfrac{1}{\left(x-x'\right)\left(y-y'\right)} \ dy' \ dx'  \right]. 
\end{align*}
Using Euler's and trigonometric formulas we get
\begin{align*}
I_0\cdot s_y^n (x,y) ={}&
 \dfrac{i}{4\pi}\mathcal{X}_{\Omega}\left(x,y\right) \left[2i \ p.v. \int\limits_{\Omega_x}\dfrac{\sin\left[\Phi\left(x,y'\right)-\Phi\left(x,y\right)\right]}{y-y'} \ dy'  \right] \\
&+
\dfrac{i}{4\pi^3}\mathcal{X}_{\Omega_x}\left(y\right)\left[2i \ p.v. \int\limits_{\Omega_y}  \ p.v. \int\limits_{\Omega_y}  \ p.v. \int\limits_{\Omega_x} \dfrac{\sin\left[\Phi\left(x',y'\right)-\Phi\left(x'',y\right)\right]}{\left(x-x'\right)\left(y-y'\right)\left(x-x''\right)} \ dy'  \ dx' \ dx''
\right] \\
={}&-
\mathcal{X}_{\Omega}\left(x,y\right)\dfrac{1}{2\pi} \ p.v. \int\limits_{\Omega_x}\dfrac{\sin\left[\Phi\left(x,y'\right)-\Phi\left(x,y\right)\right]}{y-y'} \ dy' \\
&-
\mathcal{X}_{\Omega_x}\left(y\right)\dfrac{1}{2\pi^3} \ p.v. \int\limits_{\Omega_y}  \ p.v. \int\limits_{\Omega_x}  \ p.v. \int\limits_{\Omega_y} \dfrac{\sin\left[\Phi\left(x',y'\right)-\Phi\left(x'',y\right)\right]}{\left(x-x'\right)\left(y-y'\right)\left(x-x''\right)} \ dx''  \ dy' \ dx'.
\end{align*}

This completes the proof for the sensor without modulation. \bigskip

\endgroup

Now, we derive the pyramid sensor model with circular modulation. The theoretical scheme of the non-modulated PWFS described above serves as a basis for the modulated PWFS model. The only modification to be done is to include the physical modulation of the beam
into the model:

First, physical rotation of the beam of light with a steering mirror is
represented in the theoretical model by adding a time-dependent periodic tilt \cite{BuDa06}
\begin{equation}\label{eq:periodic_tilt}
 \Phi^{mod} (x,y,t) = \alpha_{\lambda} ( x \sin (2\pi t/T) + y \cos (2\pi t/T) ) 
\end{equation} 
introducing the circular modulation path to the incoming screen $\Phi$. The constant $\alpha_{\lambda}$ denotes the modulation parameter defined in \eqref{eq:mod_param}.
Clearly, by using the non-modulated model from above, one obtains for each time step $t$ the 
non-modulated measurements $s_x^n (x,y,t), s_y^n (x,y,t)$ corresponding to the tilted phase $\Phi (x,y) + \Phi^{mod} (x,y,t)$. 

As the second step, one has to integrate these time-dependent non-modulated pyramid measurements $s_x^n (x,y,t), s_y^n (x,y,t)$
over one full time period $T$, which gives the measurements of the circularly modulated pyramid wavefront sensor as
\begin{equation} \label{eq:mod_sx_general}
 s_x^{c} (x,y) = \dfrac{1}{T} \int\limits_{-T/2}^{T/2} s_x^n (x,y,t)\  dt ,
\end{equation}
\begin{equation*} \label{eq:mod_sy_general}
 s_y^{c} (x,y) = \dfrac{1}{T} \int\limits_{-T/2}^{T/2} s_y^n (x,y,t)\ dt .
\end{equation*}

Thus, the modulated sensor measurements are described by
{\footnotesize\begin{equation*}
\begin{split}
 s_x^{c} (x,y) =& \dfrac{1}{T} \int\limits_{-T/2}^{T/2} 
\dfrac { 1 } { 2\pi } \ \mathcal{X}_{\Omega}(x,y) \ p.v. \int \limits_{\Omega_y} 
\dfrac{ \sin [ \Phi (x',y) + \Phi^{mod} (x',y,t) - \Phi (x,y) - \Phi^{mod} (x,y,t) ] } {x-x'} \ dx' \ dt\\
+& \dfrac{1}{T} \int\limits_{-T/2}^{T/2} 
\dfrac { 1 } { 2\pi^3 } \ \mathcal{X}_{\Omega_y}(x) \\ &p.v.  \int \limits_{\Omega_y}  \int \limits_{\Omega_x}
\int \limits_{\Omega_x}
\dfrac{ \sin [ \Phi (x',y') + \Phi^{mod} (x',y',t) - \Phi (x,y'') - \Phi^{mod} (x,y'',t)  ] } { (x-x') (y-y') (y-y'') } \ 
dy'' \ dy'\ dx'\ dt . \nonumber
\end{split}
\end{equation*}
}First, we want to separate the parts which depend on time to be able to integrate them,
{\footnotesize\begin{equation*}
\begin{split}
 s_x^{c} (x,y) 
=& \dfrac{1}{T} \int\limits_{-T/2}^{T/2} 
\dfrac { 1 } { 2\pi } \ \mathcal{X}_{\Omega}(x,y) \ p.v. \int \limits_{\Omega_y}
\dfrac{ \sin \left[ ( \Phi (x',y) - \Phi (x,y) ) + \left( \Phi^{mod} (x',y,t) - \Phi^{mod} (x,y,t) \right) \right] } {x-x'} \ dx'\ dt \\
+& \dfrac{1}{T} \int\limits_{-T/2}^{T/2} 
\dfrac { 1 } { 2\pi^3 } \ \mathcal{X}_{\Omega_y}(x) \\ &p.v. \int \limits_{\Omega_y}  \int \limits_{\Omega_x}
\int \limits_{\Omega_x} 
\dfrac{ \sin [ ( \Phi (x',y') - \Phi (x,y'') ) + ( \Phi^{mod} (x',y',t) - \Phi^{mod} (x,y'',t) )  ] } { (x-x') (y-y')
(y-y'') }\  dy''\ dy'\ dx'\ dt . \nonumber
\end{split}
\end{equation*}

}Note that the modulation function $\Phi^{mod}$ is linear in the first two arguments, i.e.,
{\small\begin{align} \label{eq:mod_tilts_x}
 \Phi^{mod} (x',y,t) - \Phi^{mod} (x,y,t) 
&= \alpha_{\lambda} x' \sin (2\pi t/T) + \alpha_{\lambda} y \cos (2\pi t/T)  
- \alpha_{\lambda} x \sin (2\pi t/T) - \alpha_{\lambda} y \cos (2\pi t/T)  \nonumber \\
&= \alpha_{\lambda} (x'-x) \sin (2\pi t/T) \nonumber \\
&= \Phi^{mod} (x'-x,0,t)  
\end{align}} and
{\small\begin{align} \label{eq:mod_tilts_xy}
 \Phi^{mod} (x',y',t) - \Phi^{mod} (x,y'',t) 
&= \alpha_{\lambda} x' \sin (2\pi t/T) + \alpha_{\lambda} y' \cos (2\pi t/T)  
- \alpha_{\lambda} x \sin (2\pi t/T) - \alpha_{\lambda} y'' \cos (2\pi t/T)  \nonumber \\
&= \alpha_{\lambda} (x'-x) \sin (2\pi t/T) + \alpha_{\lambda} (y'-y'') \cos (2\pi t/T) \nonumber \\
&= \Phi^{mod} (x'-x,y'-y'',t)  . 
\end{align}}
Hence, we have
{\footnotesize\begin{align*}
 s_x^{c} (x,y) 
&= \dfrac{1}{T} \int\limits_{-T/2}^{T/2} 
\dfrac { 1 } { 2\pi } \ \mathcal{X}_{\Omega}(x,y) \ p.v. \int \limits_{\Omega_y}
\dfrac{ \sin [ ( \Phi (x',y) - \Phi (x,y) ) + \Phi^{mod} (x'-x,0,t) ] } {x-x'} \ dx'\ dt  \\
&+ \dfrac{1}{T} \int\limits_{-T/2}^{T/2} 
\dfrac { 1 } { 2\pi^3 } \ \mathcal{X}_{\Omega_y}(x) \ p.v.  \int \limits_{\Omega_y} \int \limits_{\Omega_x}
\int \limits_{\Omega_x}
\dfrac{ \sin [ ( \Phi (x',y') - \Phi (x,y'') ) + \Phi^{mod} (x'-x,y'-y'',t)   ] } { (x-x') (y-y')
(y-y'') } \ dy'' \ dy'\ dx'\  dt . \nonumber
\end{align*}
}Using trigonometric formulas, we separate the time-dependent parts
{\footnotesize\begin{align*}
 s_x^{c} (x,y) 
&= \mathcal{X}_{\Omega}(x,y)\left[\dfrac { 1 } { 2\pi } \ p.v. \int \limits_{\Omega_y}
\dfrac{ \sin [ \Phi (x',y) - \Phi (x,y) ]  } {x-x'} 
\left( \dfrac{1}{T} \int\limits_{-T/2}^{T/2} 
\cos \left[ \Phi^{mod} (x'-x,0,t) \right] \ dt \right) dx' \right.\\
&+\left. \dfrac { 1 } { 2\pi } \ p.v. \int \limits_{\Omega_y}
\dfrac{ \cos [ \Phi (x',y) - \Phi (x,y) ]  } {x-x'} 
\left( \dfrac{1}{T} \int\limits_{-T/2}^{T/2} 
\sin \left[ \Phi^{mod} (x'-x,0,t) \right] \ dt \right) dx' \right] \nonumber \\
&+ \mathcal{X}_{\Omega_y}(x)\left[ \dfrac { 1 } { 2\pi^3 }  \ p.v. \int \limits_{\Omega_y}  \int \limits_{\Omega_x}
\int \limits_{\Omega_x}   
\dfrac{ \sin [ \Phi (x',y') - \Phi (x,y'') ] } { (x-x') (y-y') (y-y'') } 
\left( \dfrac{1}{T} \int\limits_{-T/2}^{T/2} 
\cos \left[ \Phi^{mod} (x'-x,y'-y'',t) \right]
\ dt \right) dy''\ dy'\ dx' \nonumber \right. \\
&+\left. \dfrac { 1 } { 2\pi^3 } \ p.v.  \int \limits_{\Omega_y}  \int \limits_{\Omega_x}
\int \limits_{\Omega_x} 
\dfrac{ \cos [ \Phi (x',y') - \Phi (x,y'') ] } { (x-x') (y-y') (y-y'') } 
\left( \dfrac{1}{T} \int\limits_{-T/2}^{T/2} 
\sin \left[ \Phi^{mod} (x'-x,y'-y'',t)\right]
\ dt \right) dy'' \ dy'\ dx' \right] . \nonumber
\end{align*}
}The second and the fourth terms equal zero, since the integrands are odd functions.
After substitution of the explicit expressions \eqref{eq:mod_tilts_x}-\eqref{eq:mod_tilts_xy} for $\Phi^{mod}$, the remaining time integrals simplify to
\begin{eqnarray*}
& & \dfrac{1}{T} \int\limits_{-T/2}^{T/2} \cos \left[ \Phi^{mod} (x'-x,y'-y'',t) \right] \ dt \\
&=& \dfrac{1}{T} \int\limits_{-T/2}^{T/2} 
\cos \left[ \alpha_{\lambda} (x'-x) \sin (2\pi t/T) + \alpha_{\lambda} (y'-y'') \cos (2\pi t/T) \right] \ dt \nonumber \\
&=& \dfrac{1}{T} \int\limits_{-T/2}^{T/2} 
\cos \left[ \alpha_{\lambda} (x'-x) \sin (2\pi t/T) \right] \cos \left[ \alpha_{\lambda} (y'-y'') \cos (2\pi t/T) \right] \ dt \nonumber \\
&-& \dfrac{1}{T} \int\limits_{-T/2}^{T/2} 
\sin \left[ \alpha_{\lambda} (x'-x) \sin (2\pi t/T) \right] \sin \left[ \alpha_{\lambda} (y'-y'') \cos (2\pi t/T) \right] \ dt  \nonumber \\
&=& \dfrac{1}{T} \int\limits_{-T/2}^{T/2} 
\cos \left[ \alpha_{\lambda} (x'-x) \sin (2\pi t/T) \right] \cos \left[ \alpha_{\lambda} (y'-y'') \cos (2\pi t/T) \right] \ dt - 0 . \nonumber \\
\end{eqnarray*}
and
\begin{eqnarray*}
\dfrac{1}{T} \int\limits_{-T/2}^{T/2} \cos \left[ \Phi^{mod} (x'-x,0,t) \right] \ dt  
&=& \dfrac{1}{T} \int\limits_{-T/2}^{T/2} \cos \left[ \alpha_{\lambda} (x'-x) \sin (2\pi t/T) \right] \ dt \\
&=& \dfrac{1}{2\pi} \int\limits_{-\pi}^{\pi} \cos \left[ \alpha_{\lambda} (x'-x) \sin (t') \right] \ dt' \\
&=& J_0[ \alpha_{\lambda} (x'-x) ] ,  \nonumber
\end{eqnarray*}
where we used the substitution $t'=2\pi t/T$ and the definition of the zero-order Bessel function 
\begin{eqnarray*}
 J_0 (x) &=& \dfrac{1}{\pi} \int\limits_0^{\pi} \cos ( x \sin t ) \ dt \\
&=& \dfrac{1}{2\pi} \int\limits_{-\pi}^{\pi} \cos ( x \sin t ) \ dt .
\end{eqnarray*}

All steps of the proof can by performed for the data $s_y^{\{c\}}$ analogously. 
\end{proof}

\bibliographystyle{plain}
\bibliography{arXiv_pyramid_iterativeI}

\end{document}